\documentclass[11pt]{article}

\usepackage{ifthen}

\newcommand{\ignore}[1]{}

\usepackage[pdfencoding=auto]{hyperref}
\usepackage[numbers]{natbib}
\usepackage{amsmath,amssymb,amsbsy,amsgen,amsopn}
\usepackage{xspace}
\usepackage{url}
\usepackage{fullpage}
\usepackage{graphicx}
\usepackage[pdfencoding=auto]{hyperref}
    \hypersetup{
	colorlinks,
	linkcolor={red!50!black},
	citecolor={blue!50!black},
	urlcolor={blue!80!black}
}
\usepackage{xspace}
\usepackage{enumerate}
\usepackage{enumitem}
\usepackage{xcolor}

\usepackage{array}
\usepackage{mathdots}
\usepackage{amsthm,amssymb,amsmath}
\usepackage{amssymb,amsmath}
\usepackage{mathtools}
\usepackage{algorithm}
\usepackage{algpseudocode}
\usepackage{mathrsfs}
\usepackage[nameinlink]{cleveref}

\newtheorem{theorem}{Theorem}
\theoremstyle{definition}
\newtheorem{definition}{Definition}
\newtheorem{lemma}[theorem]{Lemma}
\newtheorem{claim}[theorem]{Claim}

\newtheorem*{theorem*}{\bf Informal Theorem}
\newtheorem{remark}[theorem]{Remark}

\newtheorem{fact}{Fact}
\newtheorem	{itheorem}{Result}
\newcommand\tab[1][1cm]{\hspace*{#1}}

\newcommand{\Rset}{\mathbb{R}}
\newcommand{\Nset}{\mathbb{N}}

\newcommand{\cA}{\mathcal{A}}

\newcommand{\F}{\mathscr{F}}
\newcommand{\cP}{\mathcal{P}}
\newcommand{\cH}{\mathscr{H}}
\newcommand{\cI}{\mathcal{I}}
\newcommand{\cJ}{\mathcal{J}}
\newcommand{\cY}{\mathcal{Y}}
\newcommand{\cF}{\mathscr{F}}
\newcommand{\w}{\textsf{w}}

\newcommand{\from}{\leftarrow}
\newcommand{\dout}{{\delta^+}}
\newcommand{\din}{{\delta^-}}
\newcommand{\rhoout}{\rho^+}
\newcommand{\rhoin}{\rho^-}
\newcommand{\flo}{\textsf{flow}\xspace}
\newcommand{\cov}{\textsf{cov}\xspace}
\newcommand{\sink}{\textsf{sink}\xspace}
\newcommand{\val}{\textsf{val}\xspace}
\newcommand{\prob}{\textsf{prob}\xspace}
\newcommand{\rank}{\textsf{rank}_{\cF}\xspace}

\newcommand{\eps}{\varepsilon}
\newcommand{\cS}{\mathcal{S}}

\newcommand{\HS}{\textsf{Filter}\xspace}
\newcommand{\modHS}{\textsf{ModFilter}\xspace}

\newcommand{\pkco}{\textsf{P$k$CO}\xspace}
\newcommand{\lpfc}{\textsf{LP$\cF$C}\xspace}
\newcommand{\fpfc}{\textsf{FP$\cF$C}\xspace}
\newcommand{\pmco}{\textsf{PMCO}\xspace}
\newcommand{\pknapco}{\textsf{PKnapCO}\xspace}
\newcommand{\pkcoinst}{\text{$((X,d),r,k,m)$}\xspace}
\newcommand{\pmcoinst}{\text{$((X,d),r,\cF,m)$}\xspace}
\newcommand{\fpfcinst}{\text{$((X,d),r,\mu,\cF,m)$}\xspace}
\newcommand{\pknapcoinst}{\text{$((X,d),r,\w,B,m)$}\xspace}
\newcommand{\wfpp}{\textsf{WMatPP}\xspace}
\newcommand{\wkpp}{\textsf{W$k$PP}\xspace}
\newcommand{\wknappp}{\textsf{WNapPP}\xspace}
\newcommand{\mcmf}{\textsf{MCMF}\xspace}
\newcommand{\contact}{\textsf{contact DAG}\xspace}
\newcommand{\forest}{\textsf{contact forest}\xspace}

\newcommand{\total}{\textsf{total}\xspace}
\newcommand{\CovP}{\textsf{$\mathscr{P}_{\cov}$}}

\newcommand{\wkppinst}{\text{$(G=(V,E),\lambda,k)$}\xspace}
\newcommand{\wfppinst}{\text{$(G=(V,E),\lambda,X,\cY,\cF)$}\xspace}
\newcommand{\wknapppinst}{\text{$(G=(V,E),\lambda,X,\cY,\w,B)$}\xspace}

\newcommand{\pkc}{Priority $k$-Center\xspace}

\title{Revisiting Priority $k$-Center: Fairness and Outliers\footnote{A preliminary version of this work appeared in Proc.\ of ICALP 2021.}}
\author{Tanvi Bajpai\thanks{Dept.\ of Computer Science, Univ.\ of Illinois, Urbana-Champaign, Urbana,
  IL 61801. {\tt tbajpai2@illinois.edu}. Supported in part by NSF grant
CCF-1910149.}
  \and
  Deeparnab Chakrabarty\thanks{Dept.\ of Computer Science,
    Dartmouth College, Hanover NH. {\tt deeparnab@dartmouth.edu}. Supported by NSF grants CCF-1813053 and CCF-2041920.}
  \and
  Chandra Chekuri\thanks{Dept.\ of Computer Science, Univ.\ of Illinois, Urbana-Champaign, Urbana,
  IL 61801. {\tt chekuri@illinois.edu}. Supported in part by NSF grants
CCF-1910149 and CCF-1907937.}
  \and
  Maryam Negahbani\thanks{Katana Graph, {\tt maryam@katanagraph.com}. Work done as 
  a graduate student at Dartmouth College.}
}

\begin{document}



\maketitle

\begin{abstract}
  In  the {\em \pkc} problem,
  the input consists of a metric space $(X,d)$, an integer $k$ and for each
  point $v \in X$ a priority radius $r(v) > 0$. The goal is to choose
  $k$ centers $S \subseteq X$ to \emph{minimize}
  $\max_{v \in X} \frac{1}{r(v)} d(v,S)$. If all $r(v)$'s are
  uniform, one obtains the $k$-Center problem. Plesn\'ik \cite{Ples87} introduced the Priority $k$-Center problem and gave a
  $2$-approximation algorithm matching the best possible algorithm for $k$-Center.  We show how the \pkc problem is related to 
  two different notions of 
%
  \emph{fair} clustering \cite{HPST17,JKL20}.  Motivated by these developments we revisit the problem
  and, in our main technical contribution, develop a framework that
  yields constant factor approximation algorithms for Priority
  $k$-Center with \emph{outliers}.  Our framework extends to
  generalizations of \pkc to matroid and knapsack constraints, and as
  a corollary, also yields 
  algorithms with fairness guarantees in the lottery model of Harris
  et al. \cite{HLPST19}.
\end{abstract}


\section{Introduction}

Clustering is a basic task in a variety of areas, and clustering
problems are ubiquitous in practice, and are well-studied in
algorithms and discrete optimization. Recently \emph{fairness} has
become an important concern as automated data analysis and decision
making have become increasingly prevalent in society. This has motivated several
problems in fair clustering and associated algorithmic challenges.  In this paper, we
show that two different fairness views are inherently connected with
a previously studied clustering problem called the {\em \pkc} problem. 

The input to \pkc is a metric space $(X,d)$ and a priority radius
$r(v) > 0$ for each $v\in X$. The objective is to choose $k$ centers
$S\subseteq X$ such that $\max_{v \in X} \frac{d(v,S)}{r(v)}$ is
minimized.  If one imagines clients located at each point in $X$, and
$r(v)$ is the ``speed'' of a client at point $v$, then the objective
is to open $k$ centers so that every client can reach an open center
as quickly as possible. When all the $r(v)$'s are the same, then one
obtains the classic $k$-Center problem~\cite{HS86,Gon85}.
Plesn\'ik~\cite{Ples87} introduced the Priority $k$-Center problem under the name of
\emph{Weighted}\footnote{Plesn\'ik~\cite{Ples87} considered every client to
  have a weight $w(v) = \frac{1}{r(v)}$ and thus named it.  At around
  the same time, Hochbaum and Shmoys~\cite{HS86} called the version of
  $k$-Center where every {\em center} has a weight and the total
  weight of centers is bounded, the Weighted $k$-Center
  problem. Possibly to allay this confusion, G{\o}rtz and
  Wirth~\cite{GW06} called the Plesn\'ik version the \pkc
  problem. Hochbaum and Shmoys' Weighted $k$-Center problem is
  nowadays (including this paper) called the {\em Knapsack} Center
  problem to reflect the knapsack-style constraint on the possible
  centers.}  $k$-Center; the name \pkc was given by G{\o}rtz and
Wirth~\cite{GW06} and this is what we use. Plesn\'ik~\cite{Ples87}
generalized Hochbaum and Shmoys'~\cite{HS86}
$2$-approximation algorithm for the $k$-Center problem and obtained the
same bound for \pkc.
This approximation ratio is tight since $(2-\eps)$-factor approximation is ruled out even for the classic
$k$-Center problem under the assumption that $P \neq NP$~\cite{HN79}.\medskip

\noindent {\bf Connections to Fair Clustering.} Our motivation to
revisit \pkc comes from two recent papers that considered fair
variants in clustering, without explicitly realizing the connection to
\pkc. One of them is the paper of Jung, Kannan, and
Lutz \cite{JKL20} who defined a version of fair clustering as
follows. Given $(X,d)$ representing clients/people in a geographic
area, and an integer $k$, for each $v \in X$ let $r_\ell(v)$ denote
the smallest radius $r$ such that there are at least $\ell$ points of
$X$ inside a ball of radius $r$ around $v$. They suggested a notion of
fair $k$-Center as one in which each point $v \in X$ should be
served by a center not farther than $r(v) = r_{n/k}(v)$ 
since the average size of a cluster in a clustering with $k$ clusters is $n/k$.
~\cite{JKL20} describe an algorithm that finds $k$ centers such that each point $v$
is served by a center at most distance $2r_{n/k}(v)$ away from $v$.  Once the radii are fixed for the points, then one obtains an
instance of Priority $k$-Center, and their result essentially\footnote{One needs to observe that Plesn\'ik's analysis~\cite{Ples87} can be made with respect to a natural LP which has a feasible solution with $r(v) := r_{n/k}(v)$.  See~\Cref{sec:fair} for more details.} follows
from the algorithm in~\cite{Ples87}; indeed, the 
algorithm in~\cite{JKL20} is the same.
%

Another notion of fairness related to the \pkc
is the {\em lottery model} introduced by Harris et al.~\cite{HLPST19}.
In this model, every client $v\in X$ has a ``probability demand''
$\prob(v)$ and a ``distance demand'' $r(v)$. Letting $\cS$ denote all subsets of $k$ centers, their objective is to find a
distribution over $\cS$ such that for every
client $v\in X$, $\Pr_{S\sim \cS} [d(v,S) \leq r(v)]\geq \prob(v)$. An $\alpha$-approximate algorithm will either prove such a solution is not possible, or provide a
distribution where the distance to $S$ can be relaxed to
$\alpha \cdot r(v)$, i.e. $\Pr_{S\sim \cS} [d(v,S) \leq \alpha \cdot r(v)]\geq \prob(v)$. Using a standard reduction via the Ellipsoid method~\cite{CV02,AAKZ20}, this boils down to the {\em
  outlier} version of \pkc, where some
points in $X$ are allowed to be discarded.  The outlier version of
\pkc had not been explicitly studied before.\medskip

\noindent {\bf Our Contributions.} 
Motivated by these connections to fairness, 
we study the natural generalizations of
Priority $k$-Center that have been studied for the classical
$k$-center problem.  The main generalization is the {\em outlier}
version of \pkc: the algorithm is allowed to discard a certain number
of points when evaluating the quality of the centers chosen. First, the outlier version
arises in the lottery model of fairness. Second, in many situations 
it is useful and important to discard outliers to obtain a better solution.
Finally, it is also interesting from a theoretical point of view.  We also consider the
situation when the constraint on where centers can be opened is more
general than the cardinality constraint.  In particular, we study the Priority Matroid Center problem
where the set of centers must be an independent set of a given matroid, and
the Priority Knapsack Center problem where the total weight of centers
opened is at most a certain amount.  Our main contribution is
an algorithmic framework to study the outlier problems in
all these variations. 

%
%
%
%
%

\subsection{Statement of Results}
We briefly describe some variants of \pkc. In the Priority $k$-Supplier problem, the point space $X= F \uplus C$,
and the goal is to select $k$ facilities $S \subseteq F$ to minimize
$\max_{v \in C} d(v,S)/r(v)$. In the Priority Matroid Supplier problem, the subset of facilities needs to be an independent set of a matroid defined over $F$.
In the Priority Knapsack Supplier problem, each facility has a weight, and the total weight of the subset of facilities opened must be at most a certain given bound. When all $r(v)$'s are the same, each of these problems admit a $3$-approximation~\cite{HS86,CLLW13}. Our first result is that these results can be extended to the setting where the $r(v)$'s can be different. Furthermore, we establish the approximation bounds via standard LP-relaxations for the problems. Consequently, this can be used to re-derive and extend the algorithmic results in ~\cite{JKL20}; we provide details of this in ~\Cref{sec:fair}.
%

\begin{itheorem} There is a $3$-approximation for the Priority
  $k$-Supplier, Priority Matroid Supplier, and Priority Knapsack
  Supplier problems.
\end{itheorem}

%

Our second, and main technical
contribution, is a general framework to handle \emph{outliers} for priority center problems. Given
an instance of Priority $k$-Center and an integer $m \le n$, the
outlier version that we refer to as \pkco, is to find $k$ centers $S$
and a set $C'$ of at least $m$ points from $C$ such that
$\max_{v \in C'} \frac{1}{r(v)} d(v,S)$ is minimized.  While the
$k$-Center with outliers admits a clever, yet relatively simple, greedy
$3$-approximation due to Charikar et al.\ \cite{CKMN01}, a
similar approach seems difficult to adapt for Priority
$k$-Center with outliers. Instead, we take a more general and powerful LP-based
approach from \cite{CGK16,CN18} to develop a
framework to handle \pkco, and also the outlier version of Priority Matroid
Center (\pmco), where the opened centers must be an independent set of a matroid, and Priority Knapsack Center (\pknapco), where the total weight of the open centers must fit in a budget. We obtain the
following results.
\begin{itheorem}
  There is a $9$-approximation for \pkco and \pmco and
  a $14$-approximation for \pknapco. Moreover, the approximation
  ratio for \pkco and \pmco are with respect to a natural LP relaxation.
\end{itheorem}

At this point we remark that a result in Harris et al.~\cite{HLPST19}
(Theorem 2.8 in the arXiv version) also indirectly gives a
$9$-approximation for \pkco.  We believe that our framework is more
general and is able to handle \pmco and \pknapco easily. \cite{HLPST19}
does not consider these versions, and for the \pknapco
problem, their framework cannot give a constant factor approximation
since they (in essence) use a weak LP relaxation.

Furthermore, our framework yields {\em better} approximation factors when either the number of distinct priorities are small, or they are at different scales.
In practice, one probably expects this to be the case.
In particular, when there are only two distinct types of radii, we obtain a $3$-approximation which is tight; it is not too hard to show 
that it is NP-hard to obtain a better than $3$-approximation for \pkco with two types
of priorities. Interestingly, when there is a single priority, the $k$-Center with Outliers has a $2$-approximation~\cite{CGK16} showing a gap between the two problems. More details are discussed in \Cref{rem:spec}. 
We obtain $5$- and $7$-approximate solutions when the number of radii are three and four, respectively. If all the different priority values are powers of $b$ (for some parameter $b > 1$), we can derive a $\frac{3b-1}{b-1}$-approximation.
Thus, if all the priorities are in vastly different scales ($b \to \infty$), then our approximation factor approaches $3$.	 

	\noindent
	A summary of our results can be found in the third column of \Cref{tab:res}.

\begin{itheorem}
  Suppose there are only two distinct priority radii among the clients. Then
  there is a $3$-approximation for \pkco, \pmco and \pknapco. 
  With $t$ distinct types of priorities, the approximation factor for \pkco and \pmco is $2t-1$. 
  If all distinct types are powers of $b$, the approximation factor for \pkco and \pmco becomes $(3b-1)/(b-1)$.
\end{itheorem}

It is possible that the \pkco problem (without any restrictions) has a $3$-approximation, and even the natural LP-relaxation may suffice; we have not
been able to obtain a worse than $3$ integrality gap example.  
Resolving the integrality gap of the natural LP-relaxation and/or obtaining
improved approximation ratios are interesting open questions highlighted by our work.

\begin{remark}
Our results in Sections \Cref{sec:pkco,sec:pmco,sec:pknapco} are for Priority Center with Outliers problems. Our framework is able to handle the corresponding Priority Supplier with Outliers problems. We 
discuss the changes needed when handling the supplier versions in \Cref{sec:pkso}.
\end{remark}


\medskip
\noindent {\bf Consequences for Fair Clustering.}
The algorithm in \cite{JKL20} is made
much more transparent by the connection to \pkc. Since \pkc is more
general, it allows one to refine and generalize the constraints that
one can impose in the clustering model and use LP relaxations to
find more effective solutions in particular scenarios. In addition, by
allowing outliers, one can find tradeoffs between the quality of the
solution and the number of points served. We provide more
details in \Cref{sec:fair}.

Recall the lottery model of Harris et al.~\cite{HLPST19} which we
discussed previously. The algorithm in \cite{HLPST19} is based on a
sophisticated dependent rounding scheme and analysis. In fact, we
observe that implicit in their result is a $9$-approximation for \pkco
modulo some technical details.  We can ask whether the result in
\cite{HLPST19} extends to the more general setting of Matroid Center
and Knapsack Center. We prove that an $\alpha$-approximation algorithm
for weighted outliers can be translated, via the Ellipsoid method, to
yield the results in the probabilistic model of \cite{HLPST19}. This
is not surprising since a very similar reduction was shown
in~\cite{AAKZ20} in the context of the Colorful $k$-Center problem
with outliers. The advantage of this black box reduction is evident
from our algorithm from \pknapco, which is non-trivial and is based on
dynamic programming and on the round-or-cut approach since the
natural LP has an unbounded integrality gap. It is not at all obvious
how one can directly round a fractional solution to the problem while
the generic transformation is clean and simple at the high level. For
instance, our $3$-approximation for two radii extends to the lottery
model immediately.

  \begin{table}
  \begin{center}
  \caption{\label{tab:res}Traditional vs. Priority Center Approximation Results}
\begin{tabular}{|c|c|c|}
\hline
                     \textbf{Problem}  & \textbf{Traditional} & \textbf{Priority} \\ \hline
$k$-Center               & 2    \cite{HS86}       & \textbf{2}    \cite{Ples87,JKL20}, (\Cref{thm:pkc})    \\ \hline
$k$-Supplier             & 3    \cite{HS86}       & \textbf{3}     (\Cref{thm:pks})   \\ \hline
Knapsack Supplier     & 3  \cite{HS86}        &   \textbf{3} (\Cref{thm:pks})
\\ \hline
Matroid Supplier   & 3  \cite{CLLW13}         &   \textbf{3} (\Cref{thm:pks})
 \\ \hline
$k$-Center with Outliers & 2      \cite{CGK16}     & \textbf{9}   (\Cref{thm:pkco}) \\ \hline
Matroid Center with Outliers     & 3 \cite{HPST17,CN18}          &   \textbf{9} (\Cref{thm:pmco}) \\ \hline
Knapsack Center with Outliers       & 3 \cite{CN18} &   \textbf{14}  (\Cref{thm:pknapco})   \\ \hline
\end{tabular}

\end{center}

\end{table}

\subsection{Technical Discussion}\label{sec:tech-disc}
Many clustering algorithms for the $k$-Center objective use
a partitioning subroutine due to Hochbaum and Shmoys~\cite{HS86}
(HS, henceforth).  This procedure returns a partition $\Pi$ of $X$
along with a representative for each part such that all vertices of
a part ``piggy-back'' on the representative.  More precisely, if the
representative is assigned to a center $f\in X$, then so are all other
vertices in that part. To ensure a good algorithm for the
$k$-Center problem, it suffices to ensure that the radius of each part is small.

For the \pkc objective, one needs to be more careful:
to use the above idea, one needs to make sure that if vertex $v$ is
piggybacking on vertex $u$, then $r(v)$ better be more than
$r(u)$. Indeed, this can be guaranteed by running the HS procedure in a
particular order, namely by allowing vertices with smaller $r(v)$ to
form the parts first. This is precisely Plesn\'ik's
algorithm~\cite{Ples87}. In fact, this idea easily gives a $3$-approximation for the Matroid and Supplier
versions as well.
%

Outliers are challenging in the setting of \pkc. We start with the
approach of Chakrabarty et al.~\cite{CGK16} for $k$-Center with Outliers. First, they
construct an LP where $\cov(v)$ denotes the fractional coverage
(amount to which one is {\em not} an outlier) of any point, and then
write a natural LP for it. They show that if the HS algorithm is run
according to the $\cov(v)$ order (higher coverage vertices first),
then the resulting partition can be used to obtain a $2$-approximation
for the $k$-Center with Outliers problem.

When one moves to the priority $k$-Center with Outliers, one sees the
obvious trouble: what if the $r(v)$ order and the $\cov(v)$ order are
at loggerheads? Our approach out of this is a simple bucketing
idea. We first write a natural LP with fractional coverages $\cov(v)$
for every point. Then, we partition vertices into classes:
all vertices $v$ with $r(v)$ between $2^i$ and $2^{i+1}$ are in the
same class. We then  use the HS partitioning algorithm in the decreasing $\cov(v)$
order separately on each class. The issue now is to handle the interaction
across classes. To handle this, we define a directed acyclic graph across these various partitions where representative $u$ has an edge
to representative $v$ iff $d(u,v)$ is small ($\le r(u) + r(v)$). It is
a DAG because we point edges from higher $r(u)$ to the lower $r(v)$.
Our main observation is that if we can peel out $k$ paths with ``large
value'' (each representative's value is how many points piggyback on
it), then we can get a $9$-approximation for the priority $k$-center
with outlier problem. We can show that a {\em fractional} solution of
large value does exist using the fact that the DAG was constructed in
a greedy fashion. Also, since the graph is a DAG, this LP is an
integral min-cost max-flow LP. The factor $9$ arises out of a geometric
series and bucketing. Indeed, when the radii are exact powers of $2$,
we get a $5$-approximation, and when there are only two type of radii,
we get a $3$-approximation which is tight.

The preceding framework can handle the \pmco and \pknapco problems as well --- recall that these are the Outlier versions of the Priority Matroid Center and Priority Knapsack Center problems, respectively. For \pmco, the flow problem is no
longer a min-cost max-flow problem, but rather it reduces to a {\em
  submodular flow} problem which is solvable in polynomial
time. Modulo this, the above framework gives a $9$-approximation. For
\pknapco, there are two issues. One is that the flow
problem is no longer integral and
solving the underlying optimization problem is likely to be NP-hard
(we did not attempt a formal proof). Nevertheless, the framework has
sufficient flexibility. The DAG can be converted to a rooted forest on which a dynamic programming (DP) algorithm can be employed to find the desired paths; relaxing the DAG to a rooted forest amounts to an increase in the approximation factor, yielding a $14$-approximate solution. However, a second issue that we face in \pknapco is that a
fractional solution to the natural LP does not suffice when using
the DP-based algorithm on the forest; indeed the natural LP
has an unbounded gap. This issue can be circumvented by 
employing the round-or-cut approach from~\cite{CN18}; either the DP on the rooted forest succeeds
or we find a violated inequality for the large implicit LP that
we use.

\subsection{Other Related Works}
There is a huge literature on clustering, and instead of summarizing the landscape, we mention a few works relevant to our paper.
G{\o}rtz and Wirth~\cite{GW06} study the asymmetric version of the Priorty $k$-Center problem, and prove that it is NP-hard to obtain any non-trivial approximation.
A related problem to Priority $k$-Center is the {\em Non-Uniform} $k$-Center problem by Chakrabarty et al.~\cite{CGK16}, where instead of clients having radii bounds, the objective is to figure out centers 
of balls for different types of radii. Another related problem~\cite{CGS16} is the Local $k$-Median problem where clients need to connect to facilities within a certain radius, but the objective is the sum instead of the max.
%
%

Fairness in clustering has also seen a lot of works recently. Apart from the two notions of fairness described above, which can be thought of as ``individual fairness'' guarantees, Chierichetti et al.~\cite{CKLV17} introduce the ``group fairness'' notion where points have color classes, and each cluster needs to contain similar proportion of colors as in the universe. Their results were generalized by a series of follow ups~\cite{RosnerS18,BerceaGKKRS019,BeraCFN19}. A similar concept for outliers led to the study of {\em Fair Colorful $k$-Center}.
In this problem, the objective is to find $k$ centers which covers at least a prescribed number of points from each color class. This was introduced by Bandapadhyay et al.~\cite{BIPV19}, and recently true approximation algorithms were concurrently obtained by Jia et al.~\cite{JSS20} and Anegg et al.~\cite{AAKZ20}. 

Another notion of fairness is introduced by Chen et al.~\cite{ChenFLM19} in which a solution is called fair if there is no facility and a group of at least $n/k$ clients, such that opening that facility lowers the cost of all members of the group. They give a $(1+\sqrt{2})$-approximation for $L_1$, $L_2$, and $L_\infty$ norm distances for the setting where facilities can be places anywhere in the real space. Recently Micha and Shah~\cite{MS20} showed that a modification of the same approach can give a close to $2$-approximation for $L_2$ case and proved $(1+\sqrt{2})$ factor is tight for $L_1$ and $L_\infty$.

Coming back to the model of Jung et al.~\cite{JKL20}, 
the local notion of neighborhood radius is also present in the metric embedding works of~\cite{CDG06,CMM10} and were recently used by Mahabadi and Vakilian~\cite{MV20} to extend the results in~\cite{JKL20} to other objectives such as $k$-Median and $k$-Means. The Priority $k$-Median problem was further studied \cite{chakrabarty2021better,vakilian2022improved}, with \cite{vakilian2022improved} providing currentlybest known approximation.  Subsequently, Priority Matroid Median problem was studied by Bajpai and Chekuri \cite{bajpai2022bicriteria}. The  outlier versions of these problems are an open direction of study.
\section{Preliminaries} \label{sec:prelim}

We provide some formal definitions and
describe a clustering routine from~\cite{HS86}.

\begin{definition}[Priority $k$-Center]
  The input is a metric space $(X,d)$. We are also given a radius
  function $r: X \to \Rset^+$, and integer $k$. The goal is to
  find $S \subseteq X$ of size at most $k$ to minimize $\alpha$ such
  that for all $v \in X$, $d(v,S) \leq \alpha \cdot r(v)$.
\end{definition}

The following problem is an abstract generalization of the Priority $k$-Center problem, and is inspired by the corresponding generalization of $k$-Center from \cite{CN18}. This problem will be convenient when describing certain parts of our framework.

\begin{definition}[Priority $\F$-Supplier] 
  The input is a metric space $(X,d)$ where $X = F \cup C$, $C$ is the
  set of clients, and $F$ the set of facilities. We are also given a
  radius function $r: C \to \Rset^+$.  The goal is to find
  $S \subseteq F$ to minimize $\alpha$ such that for all $v \in C$,
  $d(v,S) \leq \alpha \cdot r(v)$. The constraint on $F$ is that it
  must be selected from a down-ward closed family $\F$.  Different
  families lead to different problems.  We obtain the Priority
  $k$-Supplier problem if $\F = \{S \subseteq F \mid |S| \leq k\}$.  We obtain the
  Priority Matroid Supplier problem when $(F,\F)$ is a matroid.  We
  obtain the Priority Knapsack Supplier problem when there is a weight
  function $w:F\to \Rset_{\geq 0}$ and $\F = \{S \subseteq F \mid w(S) \leq B\}$ for
  some budget $B$; here $w(S)$ denotes $\sum_{u \in S} w(u)$.
\end{definition}
\noindent
For the remainder of this manuscript, we focus on the {\em
  feasibility} version of the problem. More precisely, given an
instance of the problem, we either want to show there is no solution
with $\alpha = 1$, or find a solution with $\alpha \leq \rho$. If we
succeed, then via binary search we derive a $\rho$-approximation.

Plesn\'ik \cite{Ples87} obtained a $2$-approximation for \pkc by running a procedure similar to that of Hochbaum and Shmoys \cite{HS86}, but where points are chosen in order of priorities (\cite{HS86} uses an arbitrary order).
\Cref{alg:HS} is a slight generalization of this approach; in addition to the radius function and the metric,
it takes as input a function $\phi:X\to R_{\geq 0}$ which encodes an
ordering over the points (we can think of the points as being ordered
from largest to smallest $\phi$ values). As previously mentioned, this algorithm is a similar
procedure to that of \cite{HS86}, but while
\cite{HS86} picks points arbitrarily and \cite{Ples87} picks in order of priorities, points here get picked in the
order mandated by $\phi$.
Going forward, we use $r_u$ to denote $r(u)$ for convenience. 


\begin{algorithm}[ht]
	\caption{\HS}
	\label{alg:HS}
	\begin{algorithmic}[1]
		\Require Metric $(X,d)$, radius function $r: X \to \Rset_{>0}$, and ordering $\phi: X \to \Rset_{\geq 0}$
		\State $U \leftarrow X$ \Comment{The set of uncovered points} \label{ln:1}
		\State $S \leftarrow \emptyset$ \Comment{The set of ``representatives''}
		\While{ $U \neq \emptyset$}
	        \State $u \leftarrow \arg\max_{v\in U} \phi(v)$ \Comment{The first point in $U$ in non-increasing $\phi$ order}  \label{ln:greedy}
	        \State $S \leftarrow S \cup u$ \label{ln:rep}
	        \State $D(u) \leftarrow \{v \in U: d(u,v) \leq r_u + r_v\}$\Comment{Note: $D(u)$ includes $u$ itself} \label{ln:chld}
	        \State $U \leftarrow U \backslash D(u)$ \label{ln:remove-from-U}
		\EndWhile
		\Ensure $S$, $\{D(u) : u \in S\}$
	\end{algorithmic}
\end{algorithm}

We begin with  a few straightforward observations about the output of \HS.

\begin{fact}\label{fact:HS}
The following is true for the output of \HS:
(a) $\forall u,v \in S, d(u,v) > r_u + r_v$, 
(b) The set $\{D(u) : u \in S\}$ partitions $X$,
(c) $\forall{u \in S},\forall{v \in D(u)}, \phi(u) \geq \phi(v)$, and
(d) $\forall{u \in S},\forall{v \in D(u)}, d(u,v) \leq r_u + r_v$.
\end{fact}


Suppose we set $\phi(u) := 1/r_u$ for each $u \in X$; we obtain 
Plesn\'ik's algorithm and this yields a $2$-approximate solution for Priority $k$-Center. For completeness and later use we give a proof.

\begin{theorem}~\cite{Ples87}
\label{thm:pkc}
There is a $2$-approximation for Priority $k$-Center.
\end{theorem}


\begin{proof} 
  We claim that $S$, the output of \Cref{alg:HS} for $\phi := 1/r$, is a
  2-approximate solution; this follows from the observations
  in~\Cref{fact:HS}.  For any $v \in X$ there is some $u \in S$
  for which $v \in D(u)$. By our choice of $\phi$, $r_u \leq
  r_v$. Since $d(u,v) \leq r_u + r_v$, we have $d(u,v) \leq 2r_v$.  To
  see why $|S| \leq k$, recall that for any $u,v \in S$,
  by~\Cref{fact:HS}, $d(u,v) > r_u + r_v$ so no two points in $S$ can be covered by the same center. Thus any feasible solution needs
  at least $|S|$ many points to cover all of $S$.
\end{proof}
\noindent
In fact, under this setting of $\phi$, \Cref{alg:HS} will almost immediately gives a $3$-approximation
for Priority $\F$-Supplier for many families via the framework
in~\cite{CN18}, which we briefly describe.

The crux of the framework from \cite{CN18} is that a solution to an $\F$-Supplier problem can be determined by selecting a ``good'' partition $\Pi$ of $F$ and determining whether $\Pi$ is ``feasible" under $\F$. More formally, it requires efficient solvability of the following \emph{partition feasibility} problem: given $\Pi$, is there an $A \in \F$ such that $|A\cap P| = 1$ for
all $P\in \Pi$? If no such $A$ exists, then the original instance is infeasible. If such an $A$ does exist, then the approximation quality of $A$ can be related to the goodness of $\Pi$.

For Priority $\F$-Supplier, consider the partition $P$ returned by \Cref{alg:HS} using $\phi := 1/r$. Suppose we have partition feasibility oracle for $\F$ and it outputs a feasible $A$ for $P$. Then, by construction every $v\in X$ in part $D(u)$ satisfies $d(v,A) \leq d(u,v) + d(u,A) \leq 2r_u + r_v \leq 3r_v$ since $r_u \le r_v$. 
Furthermore, one can see that if $\Pi$ is not feasible than the original instance is not feasible. 
For the Priority $k$-Supplier, Priority Matroid Center, and
Priority Knapsack Center problems, the partition feasibility problem is solvable
in polynomial time as shown in \cite{CN18}. This leads to the following theorem.


\begin{theorem}
\label{thm:pks}
There is a $3$-approximation for Priority $k$-Supplier, Priority Knapsack Center, and the Priority Matroid Center problem.
\end{theorem}

\section{Priority $k$-Center with Outliers}\label{sec:pkco}

In this section we describe our framework for handling priorities and outliers and give a $9$-approximation algorithm for the following problem. 
%

\begin{definition}[Priority $k$-Center with Outliers (\pkco)]
The input is a metric space $(X,d)$, a radius function $r: X \to \Rset_{>0}$, and parameters $k,m \in \Nset$. The goal is to find $S \subseteq X$ of size at most $k$ to minimize $\alpha$ such that for at least $m$ points $v \in X$, $d(v,S) \leq \alpha \cdot r(v)$.
\end{definition}
\begin{theorem} \label{thm:pkco} There is a 9-approximation for \pkco. \end{theorem}
%
The following is the natural LP relaxation for the feasibility version of \pkco. For each point $v \in X$, there is a variable $0 \leq x_v \leq 1$ that denotes the (fractional) amount by which $v$ is opened as a center. We use $\cov(v)$ to indicate the amount by which $v$ is \emph{covered} by itself or other open facilities. To be precise, $\cov(v)$ is the sum of $x_u$ over all $u\in X$ at distance at most $r_v$ from $v$. Note that $\cov(v)$ is an auxiliary variable. We want to ensure that at least $m$ units of coverage are assigned using at most $k$ centers (hence the first two constraints). 
\begin{align}
\sum_{v \in X}\cov(v) &\geq m\tag{\pkco LP}\label{lp:pkco} \\
\sum_{v \in X} x_v &\leq k \notag \\
\cov(v) := \sum_{\substack{u \in X:\\ d(u,v) \leq r_v}}  x_u &\leq 1 \qquad \forall v \in X \notag \\
0 \leq x_v &\leq 1 \qquad \forall v \in X.\notag
\end{align}

Next, we define another problem called Weighted $k$-Path Packing (\wkpp) on a DAG. Our approach is to do an LP-aware reduction from \pkco to \wkpp. To be precise, we use a fractional solution of the \pkco LP to reduce to a \wkpp instance $\cJ$. We show that a good integral solution for $\cJ$ translates to a $9$-approximate solution for the \pkco instance.
We prove that $\cJ$ has a good integral solution by constructing a feasible fractional solution for an LP relaxation of \wkpp; this LP relaxation is integral.
%
Henceforth, $\cP(G)$ denotes the set of all the paths in $G$ where each path is an ordered subset of the edges in $G$.
\begin{definition}[Weighted $k$-Path Packing (\wkpp)]\label{def:wkpp}
    The input is $\cJ = \wkppinst$ where $G$ is a DAG, $\lambda: V \rightarrow \{0,1,\dots,n\}$ for some integer $n$. The goal is to find a set of $k$ \emph{vertex} disjoint paths $P	 \subseteq \cP(G)$ that maximizes:
    \begin{equation*}
        \val(P) := \sum_{p \in P} \sum_{v \in p}\lambda(v).
    \end{equation*}
\end{definition}Even though this problem is NP-hard on general graphs\footnote{$k=1$ and unit $\lambda$ is the longest path problem which is known to be NP-hard \cite{GJ79}.}, it can be easily solved if $G$ is a DAG by reducing to Min-Cost Max-Flow (\mcmf).
To build the corresponding flow network, we augment $G$ to create a new DAG $G' = (V',E')$ with a source node $s$ and sink node $t$ connected to each existing vertex, i.e. $V' = V\cup \{s,t\}$ and $E' = E \cup \{ (s,v), (v,t) \mid \forall v \in V \}$. Each node $v \in V$ has unit capacity and cost equal to $-\lambda(v)$. 
The source and sink nodes will have zero cost with capacities $\infty$ and $k$, respectively. All the arcs have unit capacity and zero cost. One can now write the \mcmf LP, which is known to be integral, for \wkpp. We use $\dout(v)$ and $\din(v)$ to denote the set of outgoing and incoming edges of a vertex $v$ respectively. The LP has a variable $y_e$ for each arc $e \in E'$ to denote the amount of (fractional) flow passing through it. Similarly, the amount of flow entering a vertex is denoted by $\flo(v) := \sum_{e \in \din(v)} y_e$. The objective is to minimize the cost of the flow which is equivalent to maximizing the negation of the costs.

\begin{align}
\max&\sum_{v \in V}\flo(v)\lambda(v) & \tag{\wkpp LP}\label{lp:wkpp}\\
\flo(v) &:= \sum_{e \in \din(v)} y_e = \sum_{e \in \dout(v)} y_e &\forall v \in V\notag\\%
\flo(t) &\leq k &\notag\\
\flo(v)  &\leq 1  &\forall v \in V \notag\\
0 \leq y_e &\leq 1  &\forall e \in E'\notag
\end{align}

\begin{claim}\label{clm:redmcmf}
\wkpp is equivalent to solving \mcmf on $G'$.
\end{claim}
\begin{proof}
Observe that any solution $P$ for the \wkpp instance translates to a valid flow of cost $-\val(P)$ for the flow problem. For any path $p \in P$ with start vertex $u$ and sink vertex $v$, send one unit of flow from $s$ to $u$, through $p$ to $v$ and then to $t$. Since the paths in $P$ are vertex disjoint and there are at most $k$ of them, the edge and vertex capacity constraints in the network are satisfied.

Now we argue that any solution to the \mcmf instance with cost $-m$ translates to a solution $P$ for the original \wkpp instance with $\val(P) = m$. To see this, note that the \mcmf solution consists of at most $k$ many $s,t$ paths that are vertex disjoint with respect to $V$. This is because of our choice of vertex capacities. Let $P$ be those paths modulo vertices $s$ and $t$. For a $v \in V$, $-\lambda(v)$  is counted towards the \mcmf cost iff $v$ has a flow passing through it which means $v$ is included in some path in $P$. Thus $\val(P) = m$.
\end{proof}
%

	%

\subsection{Reduction to \wkpp}\label{subsec:redwkpp}
 Using a fractional solution of the \pkco LP we construct a \wkpp instance. In particular, we use the \cov assignment generated by the LP solution.  
Without loss of generality, by scaling the distances, we assume that the smallest neighborhood radius is 1. Let $t := \lceil \log_2 r_{max} \rceil$, where $r_{max}$ is the largest value of $r$ (after scaling). We use $[t]$ to denote $\{1,2,\dots,t\}$. Partition $X$ according to each point's radius into $C_1,\dots,C_t$, where $C_i := \{ v \in X: 2^{i-1} \leq r_v < 2^i\}$ for $i \in [t]$. Note that some sets may be empty if no radius falls within its range. 

\Cref{alg:pkco} shows the \pkco to \wkpp reduction. The algorithm constructs a DAG called \contact (see \Cref{def:contactpkco}) as a part of the \wkpp instance definition. We first run \Cref{alg:HS} on each $C_i$, where $\phi := \cov$, to produce a set of representatives $R_i$ and their respective clusters $\{ D(u) : u \in R_i \}$. The $\lambda$ values are constructed using each $D(v)$. Each $R_i$ defines a \emph{row} of the \contact starting with $R_t$ at the top. Arcs in the \contact exist only between points in different rows, and only when there is a point in $X$ that can cover them both within their desired radii (a more precise definition is given in \Cref{def:contactpkco}). We always have arcs pointing downwards, that is, from points in $R_i$ to points in $R_j$ where $i > j$. See \Cref{fig:cdag} for an example on how a \contact looks like.

\begin{algorithm}[!ht]
	\caption{Reduction to \wkpp}
	\label{alg:pkco}
	\begin{algorithmic}[1]
		\Require \pkco instance $\cI = $\pkcoinst and assignment $\{\cov(v) \in \Rset_{\geq 0} :v\in X\}$
		\State $R_i, \{D(u): u\in R_i\} \from \HS((C_i, d), r, \cov)$ for all $i \in [t]$ \label{ln:HS}
		\State \text{Construct \contact $G = (V, E)$ per \Cref{def:contactpkco}} 
		\State $\lambda(v) \from |D(v)|$ for all $v \in V$ \label{ln:lambda}
		\Ensure \wkpp instance $\cJ = \wkppinst$
	\end{algorithmic}
\end{algorithm}


\begin{definition}[\contact]\label{def:contactpkco}
Let $R_i \subseteq C_i$, $i \in [t]$ be the set of representatives acquired after running \HS on $C_i$ according to Line 1 of \Cref{alg:pkco}. \contact $G = (V,E)$ is a DAG on vertex set $V = \bigcup_i R_i$ where the arcs are constructed by the following rule:
\begin{align*}
    \textnormal{For $u \in R_i$ and $v \in R_j$ where $i > j$ }, (u,v) \in E  \iff \exists f \in X: d(u,f) \leq r_u \textnormal{ and } d(v,f) \leq r_v.
\end{align*}

\end{definition}
 \begin{figure}
     \centering
      \includegraphics[width=0.9\textwidth]{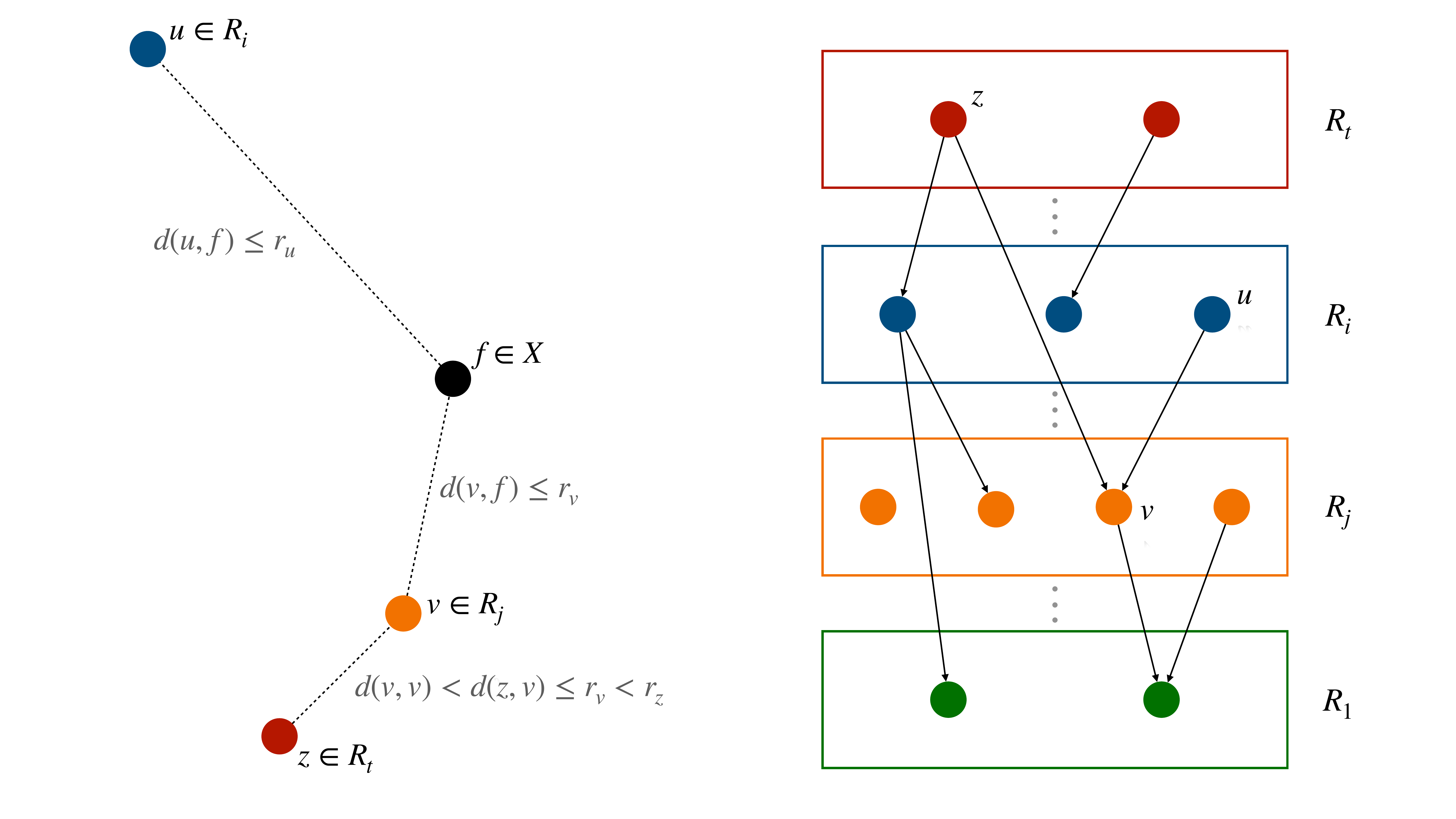}
    \caption{An example \contact. Arcs are constructed per \Cref{def:contactpkco}.}
     \label{fig:cdag}
\end{figure}

Our first observation is that the \wkpp instance has a good fractional solution, and since the LP is integral, it also has a good integral solution. The proof of this claim, i.e. \Cref{lma:lp-pkco}, is based on standard network flow ideas.

\begin{lemma}\label{lma:lp-pkco}
	There is a valid solution to \ref{lp:wkpp} of value at least $m$ for the \wkpp instance $\cJ$.
	Since \ref{lp:wkpp} is integral, this implies $\cJ$ has an integral solution of value at least $m$.
\end{lemma}

\begin{proof} 
	We construct a fractional solution $y$ for the \wkpp LP with objective value at least $m$. 
Recall $V = \bigcup_i R_i$ per definition of \contact. For any $f \in X$ let $A_f := \{v \in V \mid d(f,v) \leq r_v\}$ be the set of points $v \in V$ for which $x_f$ contributes to $\cov(v)$. That is for all $v \in V$:
\begin{equation}\label{eq:deltacov}
    \sum_{\substack{f \in X:\\ v \in A_f}} x_f = \cov(v).
\end{equation}
There is an edge between points $u,v \in V$ iff there exists some $f \in X$ for which $d(u,f) \leq r_u$ and $d(v,f) \leq r_v$. Thus for any $u,v \in A_f$, we have $(u,v) \in E$. Now recall\footnote{Refer to the definition of \wkpp and the LP formulation based on \mcmf.} how we augmented $G$ to get a flow network $G'$ by adding source and sink vertices $s$ and $t$, plus arcs $(s,v)$ and $(v,t)$ for all $v \in V$. Observe that $A_f$ resembles an $s,t$ path in $G'$. Formally, let $(u_1,\dots,u_l)$ be $A_f$ sorted in decreasing order of neighborhood radii.  Define $p_f$ to be the $s,t$ path that passes through $A_f$ in this order. That is, $p_f = ((s,u_1),(u_1,u_2),\dots (u_l,t))$. Note that the same arc can be in $p_f$ of multiple $f \in X$. This motivates the definition of $H_e := \{f \in X\mid e \in p_f\}$. Set $y$ as follows:
\begin{equation*}
    y_e := \sum_{f \in H_e} x_f.
\end{equation*}
Now we argue that $y$ is a feasible solution for \ref{lp:wkpp}. First, notice that the flow is conserved for each vertex $v \in V$. That is, $\sum_{e \in \din(v)} y_e = \sum_{e \in \dout(v)} y_e$. This is due to the fact that for any $f \in X$, we add the same amount $x_f$ to $y_e$ of all $e \in p_f$. Next, we see that $\flo(v) = \cov(v)$.
\begin{equation}\label{eq:flocov}
    \flo(v) = \sum_{e \in \din(v)} y_e = \sum_{e \in \din(v)}\sum_{f \in H_e} x_f = \sum_{\substack{f \in X:\\ v \in A_f}} x_f = \cov(v).
\end{equation}
The constraint $\flo(v) \leq 1$ and $\flo(t) \leq k$ in \ref{lp:wkpp} follow from similar constraints in \ref{lp:pkco}: The former is due to the constraint $\cov(v) \leq 1$ and the latter is by $\sum_{f \in X} x_f \leq k$.

The last thing to verify is that the value of the \ref{lp:wkpp} objective for this solution is at least $m$. That is, is $\sum_{v \in V}\flo(v)\lambda(v) \geq m$. We start with the constraint $m \leq \sum_{v \in X}\cov_v $ from \ref{lp:pkco}:
\begin{align*}
    m \leq \sum_{v \in X}\cov(v) &= \sum_{u \in V}\sum_{v \in C(u)} \cov(v) & \text{(by \Cref{fact:Vpart})}\\
    &\leq \sum_{u \in V}|D(u)|\cov(u)& \text{(by~\Cref{fact:HS})}\\
    & = \sum_{u \in V}|D(u)|\flo(u) = \sum_{u \in V}\flo(u)\lambda(u). &\text{(By \eqref{eq:flocov} and Definition of $\lambda$)}
\end{align*}
\end{proof}

\Cref{thm:pkco} now follows from the following lemma.

\begin{lemma}\label{lma:pkco} 
	Any solution with value at least $m$ for the \wkpp instance $\cJ$ given by \Cref{alg:pkco} translates to a 9-approximation for the \pkco instance $\cI$.
\end{lemma}

Before proving \Cref{lma:pkco}, we begin with a few observations. Per the definition of \contact we have the following property.

\begin{fact}\label{fact:arcdist}
If $u \in R_i$, $v \in R_j$, and $(u,v)$ is an arc in \contact, $d(u,v) \leq r_u + r_v$. 
\end{fact}

 Note that the converse of \Cref{fact:arcdist} does not necessarily hold. The next fact follows straightforwardly from \Cref{fact:HS}; by construction, $\{C_i\}_{i \in [t]}$ partitions $X$ and \HS further partitions each $C_i$ according to~\Cref{fact:HS}(b).

\begin{fact}\label{fact:Vpart}
$\{D(v), v \in V\}$ as constructed in \Cref{alg:pkco} partitions $X$.
\end{fact}
In the following claim, we bound the distances between points on a \contact, and therefore are bounding the distances between points on any path that could be returned for the \wkpp instance. 

\begin{claim}\label{clm:dist2root}
For any $u \in R_i$, $v \in R_j$ reachable from $u$ in a \contact, $d(u,v) < 3\cdot 2^i$.
\end{claim}
\begin{proof}
Observe that by definition of \contact, $i > j$. A path from $u$ to $v$ may contain a vertex from any level of the DAG between $i$ and $j$. In the worst case, the path has a vertex $w_k$ from every level $R_k$ for $ j < k < i$:
\begin{align*}
d(u,v) &\leq d(u,w_{i-1}) + d(w_{i-1},w_{i-2}) + \hdots + d(w_{j+1},v) &\\
&\leq (r_u + r_{w_{i-1}}) + (r_{w_{i-1}} + r_{w_{i-2}}) + \hdots + (r_{w_{j+1}} + r_v) &\text{(by \Cref{fact:arcdist})}\\ 
&= r_u + 2 \sum_{k = j + 1}^{i-1} r_{w_k} + r_v < r_u + 2\sum_{k = 1}^{i-1} 2^k & \text{(by definition of $w_k \in C_k$)}\\ 
&= r_u + 2\cdot (2^i-2) < 3\cdot2^i. &\text{($r_u < 2^i$)}
\end{align*}
Note the slack which leads to the first strict inequality; indeed, we would still have this inequality if there was an extra $r_v$ in the summation.
This slack is utilized when we move to more general versions of the problem, in particular, in~\Cref{clm:dist2root2S} and in~\Cref{sec:pkso}.
\end{proof}

Now we are armed with all the facts needed to prove \Cref{lma:pkco}.

\begin{proof}[Proof of \Cref{lma:pkco}]
We are assuming the constructed \wkpp instance has a solution of value at least $m$, which means there exists a set of $k$ disjoint paths $P \subseteq \cP(G)$ in the \contact such that $\val(P)  \geq m$. For any path $p \in P$, let $\sink(p)$ denote the last node in this path (i.e. $\sink(p) = \arg\min_{u \in p} r_u$). Our final solution would be $S := \{ \sink(p) : p \in P\}$. We argue that this $S$ is a 9-approximate solution for the initial \pkco instance. Since $P$ has at most $k$ many paths, $|S| \leq k$. 

Now we show any $w \in D(u)$ where $u \in p \in P$, can be covered by $v = \sink(p)$ with dilation at most 9. Assume $u \in R_i$ for some $i \in [t]$.
\begin{align*}
    d(w,v) &\leq d(w,u) + d(u,v) < r_w + r_u + 3\cdot 2^i  &\text{(by~\Cref{fact:arcdist} and \Cref{clm:dist2root}}) \\
    &< r_w + 4\cdot 2^i \leq 9r_w. &\text{($r_u < 2^i$ and $2^{i-1} \leq r_w$)}
\end{align*}


The last piece is to argue at least $m$ points will be covered by $S$. The set of points that are covered by $S$ within $9$ times their radius is precisely the set $D_{\total} := \bigcup_{p \in P}\bigcup_{u \in p} D(v)$. So we need to show $|D_{\total}| \geq m$. By \Cref{fact:Vpart} we have:
\begin{equation*}
    |D_{\total}| = |\bigcup_{p \in P}\bigcup_{v \in p} D(v)| = \sum_{p \in P}\sum_{v \in p} |D(v)| = \val(P),
\end{equation*}
where the last equality is by the definition of $\lambda(v), v \in V$ (in Line 3 of \Cref{alg:pkco}) and definition of $\val(P)$. By assumption $\val(P) \geq m$ thus $D_{\total}$ contains at least $m$ points.
\end{proof}
\noindent

As previously stated, \Cref{lma:pkco} along with \Cref{lma:lp-pkco} together complete the proof of \Cref{thm:pkco}.

\begin{remark}\label{rem:spec}
In the special case where there are 2 types of radii, we can slightly modify our approach to get a 3-approximation algorithm. This result is tight. 
To see this consider \pkco instances where clients having priority radii of either $0$ or $1$, with $n_0$ of the former type and $n_1$ of the latter, and the number of outliers allowed is $n_0 - k$.
Clients with priority radii $0$ either need to have a facility opened at that same point, or need to be an outlier. Since only $n_0 - k$ outliers and $k$ centers are allowed, all the outliers and centers are on these $n_0$ points. Thus, the $n_0$ points act as facilities in the $k$-Supplier problem which is hard to approximate with a factor better than 3. This shows a gap with the $k$-Center with outliers, which has a $2$-approximation~\cite{CGK16}. 

\end{remark}

In general, our framework yields improved approximation factors when the number of distinct radii is less than 5 (see \Cref{thm:2pkc}). In the special case when all input radii are already powers of $2$ (foregoing the loss incurred from bucketing), our algorithm is actually a $5$-approximation. This factor improves if the radii are powers of some $b > 2$ and approaches $3$ as $b$ goes to infinity (see \Cref{thm:bpkc}).

\begin{theorem}
\label{thm:2pkc}
There is a $(2t-1)$-approximation for \pkco instances where there are only $t \geq 2$ types of radii.
\end{theorem}
\begin{proof}
Given \pkco instance $\cI$ with $t \geq 2$ distinct types of radii, $r_1 < r_2 < \dots < r_t$, obtain a fractional solution $x$ by solving the \pkco LP. Partition $X$ according to each point's radius into $C_1,\dots,C_t$, where $C_i$ is points of radius $r_i$ for $i \in [t]$. Construct a \wkpp instance $\cJ$ by running \Cref{alg:pkco} with input $\cov$ corresponding to $x$. Assuming the \wkpp instance has a solution with value at least $m$, we can show how to obtain a $(2t-1)$-approximate solution as follows. Let $P$ be the \wkpp solution. Take any $p \in P$. If $p$ is a single vertex, simply add it to solution $S$. Otherwise, let $v = \sink(p)$ where $v \in R_j$ for some $j \in [t]$, and $v'$ is the vertex before $v'$ in $p$, i.e. $v' \in p$ where $v' \in R_{j'}$ for some $j' \in [t]$ such that $j' > j$, and $p$ contains the directed edge $(v',v)$. Now, instead of adding the point $v$ to $S$, as was done in the proof of \Cref{lma:pkco}, we instead add the point $f \in X$ that covers both $v$ and $v'$ within distance $r_j$ and $r_{j'}$, respectively. Such an $f$ exists per \Cref{def:contactpkco}. 

For points $w \in D(v)$ and so $r(w) = r_j$, we have $d(w,f) \leq d(w,v) + d(v,f) \leq 2r_j + r_j = 3r_j$, and thus $w$ is covered by $f$ with dilation at most $3 \leq 2t-1$ since $t\geq 2$.
The exact same argument holds if $w\in D(v')$ and $r(w) = r_{j'}$. For any $w \in D(u)$ where $u \in p$ and $u \in R_i$ for some $i \in [t]$ such that $i > j'$, we can bound the distance between $u$ and $v'$ similar to the proof of \Cref{clm:dist2root}. However, since we do not bucket radii values by powers of $2$, we instead bound the radius of any vertex between $u$ and $v'$ by $r_i$. Using this, and the fact that $j' > j$ implies that $j' \geq 2$, we can derive that \[d(u,v') \leq r_i + 2 \sum_{k = j'+1}^{i-1} r_k + r_{j'}  < 2\sum_{k=3}^i r_i = 2(i-2)r_i.\] 
Recall that $d(w,u) \leq 2r_i$ (\Cref{fact:HS}) and $d(v',f) \leq r_{j'} \leq r_i$. Using this and triangle inequality, we can conclude that \[ d(w,f) \leq d(w,u) + d(u,v') + d(v',f) < 3r_i + 2(i-2)r_i = (2i-1)r_i. \]
 Thus, $w$ will be covered by $f$ with dilation at most $2i-1 \leq 2t-1$. 

To argue at least $m$ points will be covered by $S$, follow the analogous argument from the proof of \Cref{lma:pkco}. The remainder of this proof, i.e. showing that $\cJ$ does indeed have a solution of value at least $m$ that can be determined in polynomial time using an \mcmf algorithm, is identical to the proof of \Cref{lma:lp-pkco}. 
\end{proof}
\begin{theorem}
\label{thm:bpkc}
There is a $((3b-1)/(b-1))$-approximation for \pkco instances where the radii are powers of $b \geq 2$.
\end{theorem}
\begin{proof}
Given \pkco instance $\cI$ obtain fractional solution $x$ by solving the \pkco LP. Partition $X$ according to each point's radius into $C_1,\dots,C_t$, where $t := \lceil \log_b r_{max} \rceil$ and $C_i := \{ v \in X \mid r_v = b^{i-1}\}$ for $i \in [t]$. Construct a \wkpp instance $\cJ$ by running \Cref{alg:pkco} with input $\cov$ corresponding to $x$. Assume the \wkpp instance has a solution $P$ with value at least $m$. For any $p \in P$ add $v = \sink(p)$ to solution $S$. Consider arbitrary $w \in D(u)$ where $u \in p \in P$ and assume $u \in R_i$ for some $i \in [t]$. Similar to the proof of \Cref{clm:dist2root} one can show $d(u,v) < ((b+1)/(b-1))\times b^{i-1}$. By~\Cref{fact:HS} $d(w,u) \leq r_w + r_u = 2b^{i-1}$. Thus any $w$ is covered by dilation $(3b-1)/(b-1)$ as $d(w,v) \leq d(w,u) + d(u,v) < 2b^{i-1} + ((b+1)/(b-1))\times b^{i-1} = (3b-1)/(b-1) r_w$. 

As in the proof of \Cref{thm:2pkc}, the remaining pieces of this proof will follow the analogous arguments from the proofs of \Cref{lma:pkco} and \Cref{lma:lp-pkco}. 
\end{proof}



Before proceeding to the next sections, we make the following observation which will help us in handling more general constraints on the choice of facilities. 
The only place we used the cardinality constraint on the facilities (i.e. $\sum_{f \in X} x_f \leq k$) is to make sure that the solution $y$ corresponds to at most $k$ many paths. The reduction to the path problem and the analysis of the approximation ratio do not dependent on the specifics of the constraint.

\section{Priority Matroid Center with Outliers}\label{sec:pmco}
In this section, we generalize the results from the previous section to the Priority Matroid Center with Outliers problem.

\begin{definition}[Priority Matroid Center with Outliers (\pmco)]
The input is a metric space $(X,d)$, parameter $m \in \Nset$, radius function $r: X \to \Rset_{>0}$, and $\cF \subseteq 2^X$ a family of independent sets of a matroid. The goal is to find $S \in \cF$ to minimize $\alpha$ such that for at least $m$ points $v \in X$, $d(v,S) \leq \alpha \cdot r(v)$.
\end{definition}

\begin{theorem} \label{thm:pmco} There is a
	9-approximation for \pmco. \end{theorem}
	
As in the previous section, we assume $\alpha = 1$ and consider the feasibility version of the problem. For any $S \subseteq X$, let $\rank(S)$ be the rank of $S$ in the given matroid. 
The natural LP relaxation for this problem is very similar to that of \ref{lp:pkco} except that we replace the cardinality constraints with \emph{rank constraints} $x(S) \leq \rank(S)$ for all $S \subseteq X$. This is because for any $S \in \cF$, $|S| = \rank(S)$.

\begin{align}
\sum_{v \in X}\cov(v) &\geq m\tag{\pmco LP}\label{lp:pmco} \\
\sum_{v \in S} x_v &\leq \rank(S)& \forall S \subseteq V \notag \\
\cov(v) := \sum_{\substack{u \in X:\\ d(u,v) \leq r_v}}  x_u &\leq 1 & \forall v \in X \notag \\
0 \leq x_v &\leq 1 & \forall v \in X.\notag
\end{align}

Similar to \wkpp, the path packing version of \pmco defined below. 
Recall from the previous section, that after reducing from \pkco to \wkpp we returned a set of $k$ vertices in DAG $G$ as our final solution. Now that we have matroid constraints, we must instead return a set $S$ of vertices such that $S \in \cF$. Doing so is not entirely straightforward, since our reduction does not guarantee that such a subset of vertices actually exists and covers enough points in their corresponding vertex disjoint paths. Instead, we show there is an $S \in \cF$ such that each member of this $S$ is \emph{close} to some vertex of $G$. These close points in $G$ will correspond to a set of vertex disjoint paths that will cover enough points. 

\begin{definition}[Weighted $\cF$-Path Packing]\label{def:wfpp}
   The input is a DAG $G=(V,E)$, $\lambda: V \rightarrow \mathbb{Z}_+$, a finite set $X$, a set $\cY = \{ Y_v \subseteq X\mid v \in V\}$, and a down-closed family $\cF \subseteq 2^X$ of independent sets. The goal is to find a set of vertex disjoint paths $P \subseteq \cP(G)$ in $G$ with maximum $\val(P):= \sum_{p \in P} \sum_{v \in p}\lambda(v)$ with the following constraint: there exists a set $S \in \cF$ such that $\forall p \in P$, $S \cap Y_{\sink(p)} \neq \emptyset$ where $\sink(p)$ is the last vertex of the path $p$. 
\end{definition}

When $\cF$ is the collecton of independent set of a matroid on $X$, 
we refer to the preceding problem as the Weighted Matroid Path Packing problem (\wfpp).

In the preceding definition, the collection of sets $\cY$ is meant to describe the set of points in $X$ that a vertex in $G$ can be considered close to. As mentioned previously, such a set is needed because the selected paths of $G$ need to correspond to a set from $\cF$. Notice that per this definition, $\cF$ has no requirement besides it being a family of independent sets of a matroid described on $X$, and $X$ could be separate from $G$ (though, for our reduction, we do ultimately construct $G$ with respect to our \pmco input point set $X$). It may be helpful to notice that $\wfpp$ is a generalization of $\wkpp$ from the previous section (\Cref{def:wkpp}): if $V \subseteq X$ and the matroid from which $\cF$ is described is a uniform matroid, i.e. $\cF$ contains all subsets $S \subseteq X$ such that $|S| \leq k$, then collection $\cY$ can simply be set to $\{ \{v\} \mid v \in V\}$, resulting in an instance equivalent to $\wkpp$.

Now, observe that the reduction procedure in \Cref{alg:pkco} and all of our subsequent observations in \Cref{subsec:redwkpp} do not rely on how we define a feasible set of centers. Hence, the main obstacle in proving \Cref{thm:pmco} lies in our reduction to \mcmf. Luckily, the result of \cite{ckrv12} helps us address this by giving LP integrality results similar to \mcmf using the following formulation on directed \emph{polymatroidal flows}~\cite{EG77,hassin1982,LM82}: For a network $G'=(V',E')$, for all $v \in V'$, we are given polymatroids\footnote{Monotone integer-valued submodular functions.} $\rhoin_v$ and $\rhoout_v$ on $\din(v)$ and $\dout(v)$ respectively. For every arc $e \in E'$ there is a variable $0 \leq y_e \leq 1$. The capacity constraints for each $v \in V'$ are defined as:

\begin{align*}
    \sum_{e \in U} y_e &\leq \rhoin_v(U) \qquad \forall U \subseteq \din(v)\\
    \sum_{e \in U} y_e &\leq \rhoout_v(U) \qquad \forall U \subseteq \dout(v).
\end{align*}

We augment the DAG $G$ given in \wfpp to construct a polymatroidal flow network $G'$. In this new network, $V' = V \cup X\cup \{s,t\}$ (see \Cref{rem:xv}) where each node $v \in V$ has cost $-\lambda(v)$. $E'$ includes all of $E$, plus arcs $(s,v)$ for all $v \in V$. Finally, instead of adding arcs $(v,t)$, we add arcs $(v,f)$ and $(f,t)$ for all $f \in Y_v$. 
\begin{remark}\label{rem:xv}
Even though a vertex $v \in V$ might correspond to a point in $X$, in $V'$ we make a distinction between the two copies. 
\end{remark}

The polymatroids for this instance are constructed as follows: for any $v \in V \cup X$, $\rhoin_v(U) = 1$ for all non-empty $U \subseteq \din(v)$ and $\rhoout_v$ is defined similarly on $\dout(v)$. For $s$, we only have outgoing edges where $\rhoout_s(U) = |U|$ for all $U \subseteq \dout(s)$. Finally, we enforce the matroid constraints of $\cF$ on $t$. For any $U \subseteq \din(t)$, let $T \subseteq X$ be the set of starting nodes in $U$. That is, $U = \{(f,t): f \in T\}$. Set $\rhoin_t(U) = \rank(T)$. Since $\din(t) \subseteq X$, these capacity constraints on $t$ are equivalent to the following set of constraints:
\begin{equation*}
    \sum_{f \in T} y_{(f,t)} \leq \rank(T) \qquad \forall T \subseteq X: \{(f,t): f\in T\} \subseteq \din(t).
\end{equation*}

Now, we prove a claim analogous to that of \Cref{clm:redmcmf}.

\begin{claim}\label{clm:redmat}
\wfpp is equivalent to solving the polymatroidal flow on network $G'$.
\end{claim}
\begin{proof}
Any solution $P$ for the \wfpp instance translates to a valid flow of cost $-\val(P)$ for the flow problem. Let $S \in \cF$ be the independent set that intersects $Y_{\sink(p)}$ for all $p \in P$. For any path $p \in P$ with start vertex $u$ and sink vertex $v$, take arbitrary $f \in S \cap Y_v$. Send one unit of flow from $s$ to $u$, through $p$ to $v$ and then to $f$ and $t$. All the polymatroidal constraints in \ref{lp:wfpp} are satisfied.

Now we argue that any solution to the flow instance with cost $-m$ translates to a solution $P$ for \wfpp with $\val(P) = m$. To see this, note that the flow solution consists of $s,t$ paths that are vertex disjoint with respect to $V\cup X$. This is due to our choice of $V \cup X$ polymatroids. Each path passes through one $v \in V$, then immediately to $f \in Y_v$ and then ends in $t$. By polymatroidal constraints on $t$, the subset of $X$ that has a flow going through it will be an independent set of $\cF$.

Let $P$ be the described paths induced on $V$. For $v \in V$, $-\lambda(v)$  is counted towards the \mcmf cost iff $v$ has flow passing through it. This means $v$ is included in some path in $P$. Thus $\val(P) = m$.
\end{proof}

The polymatroidal LP for this particular construction is as follows (recall $\flo(v) := \sum_{e \in \din(v)} y_e$):
\begin{align}
\max&\sum_{v \in V}\flo(v)\lambda(v) & \tag{\wfpp LP}\label{lp:wfpp}\\
\flo(v) &:= \sum_{e \in \din(v)} y_e = \sum_{e \in \dout(v)} y_e &\forall v \in V \cup X\notag\\%
\sum_{f \in T} y_{(f,t)} &\leq \rank(T) &\forall T \subseteq X: \{(f,t): f\in T\} \subseteq \din(t)\notag\\
\flo(v)  &\leq 1  &\forall v \in V \cup X \notag\\
0 \leq y_e &\leq 1  &\forall e \in E'\notag
\end{align}
By \cite{ckrv12}, \ref{lp:wfpp} is integral and there are polynomial time algorithms to solve it.

For our reduction from \pmco to \wfpp, most of the notation and results can be recycled from \Cref{subsec:redwkpp}. Specially, the reduction itself (\Cref{alg:pmco}) is just \Cref{alg:pkco} along with the explicit construction of the sets of $\cY$ (\Cref{lin:yv}). 

\begin{remark}\label{rem:arc}
Per the definition of arc in \contact, as well as our setting of $Y_v$ given in \Cref{lin:yv} of \Cref{alg:pmco}, for two nodes $u,v \in V$, $(u,v)$ is an arc iff $Y_v$ intersects $Y_u$.
\end{remark}

\begin{algorithm}[!ht]
	\caption{Reduction to \wfpp}
	\label{alg:pmco}
	\begin{algorithmic}[1]
		\Require \pmco instance $\cI = $\pmcoinst and assignment $\{\cov(v) \in \Rset_{\geq 0} :v\in X\}$
		\State $R_i, \{D(u): u\in R_i\} \from \HS((C_i, d), r, \cov)$ for all $i \in [t]$ 
		\State \text{Construct \contact $G = (V, E)$ per \Cref{def:contactpkco}} 
		\State $\lambda(v) \from |D(v)|$ for all $v \in V$ \label{lin:lam}
		\State $Y_v \from \{u \in X: d(u,v) \leq r_v\}$ for all $v \in V$ \label{lin:yv}
		\Ensure \wfpp instance $(G=(V,E),\lambda,X,\cY = \{Y_v \mid v \in V\},\cF)$
	\end{algorithmic}
\end{algorithm}

Before we begin proving our 9-approximation result for \pmco, we need to slightly modify \Cref{clm:dist2root} to account for the fact that a vertex covered by  $v$ (the sink of some path) has to travel slightly farther than $v$ to reach an $f \in Y_v$. Fortunately, the proof of \Cref{clm:dist2root} has a slight slack that allows us to derive the same distance guarantees even with this extra step.

\begin{claim}\label{clm:dist2root2S}
For any $u \in R_i$ and $v \in R_j$ reachable from $u$ in \contact $G$, and any $f \in Y_v$, $d(u,f) < 3\cdot 2^i$.
\end{claim}
\begin{proof}
By definition of $G$ it must be the case that $i > j$. Also for all $f \in Y_v$, $d(f,v) \leq r_v$. If $v$ is reachable from $u$, a path between $u$ and $v$ may contain a vertex $w_k$ from every level $R_k$ for $ j < k < i$:
\begin{align*}
d(u,f) & \leq d(u,v) + d(v,f) \leq d(u,v) + r_v &\\
&\leq (r_u + r_{w_{i-1}}) + (r_{w_{i-1}} + r_{w_{i-2}}) + \hdots + (r_{w_{j+1}} + r_v) + r_v&\text{(by \Cref{fact:arcdist})}\\ 
&\leq r_u + 2\sum_{k = 1}^{i-1} 2^k &\text{(by definition of $C_k$)}\\ 
&= r_u + 2\cdot (2^i-2) < 3\cdot2^i. &\text{($u \in C_i$, $r_u < 2^i$)}
\end{align*}
\end{proof}
Since the previous claim has the same guarantee as \Cref{clm:dist2root}, \Cref{lma:pkco} easily translates to the following: 

\begin{lemma}\label{lma:pmco} Any solution with value at least $m$ for the output of \Cref{alg:pmco} translates to a 9-approximation for the input $\cI$.
\end{lemma}
\begin{proof}
Let $P \subseteq \cP(G)$ be the promised \wfpp solution. Let $S \in \cF$ be the independent set that intersects $Y_{\sink(p)}$ for all $p \in P$. By \Cref{clm:dist2root2S}, $S$ covers all the vertices $v \in V$ that are included in $P$ by dilation 3. Proof of \Cref{lma:pkco} shows that for any such $v \in V$ covered by $P$ and any $w \in D(v)$, $d(w,S) \leq 9r_w$. This holds for at least $m$ points.
\end{proof}
We can now prove our 9-approximation result for \pmco.
\begin{proof}[Proof of \Cref{thm:pmco}]
The algorithm is very similar to that of \Cref{thm:pkco}: 
Given \pmco instance $\cI = $\pmcoinst, solve the \ref{lp:pmco} and use the solution in the procedure of \Cref{alg:pmco} to reduce to \wfpp instance $\cJ = \wfppinst$. Let $P \subseteq \cP(G)$ be the solution to this instance and $S \in \cF$ be the independent set that intersects $Y_{\sink(p)}$ for all $p \in P$. If $\val(P) \geq m$, $S$ is a 9-approximate solution for $\cI$ via \Cref{lma:pmco}. So we prove such solution $P$ exists by constructing a feasible (possibly fractional) \ref{lp:wfpp} solution.

Take the \contact from \Cref{alg:pmco}, $G=(V,E)$, and recall that each $v \in V$ is also a point in $X$. For any $f \in X$ let $A_f := \{v \in V: d(f,v) \leq r_v\}$ be the set of points $v \in V$ for which $x_f$ contributes to $\cov(v)$. By definition of an edge in \contact, for any $u,v \in A_f$, we have $(u,v) \in E$. Define $p_f$ to be the $s,t$ path that passes through $A_f$ in the order of decreasing neighborhood radii. Formally, let $(u_1,\dots,u_l)$ be $A_f$ sorted in decreasing order of neighborhood radii. Then, $p_f = ((s,u_1),(u_1,u_2),\dots (u_l,f),(f,t))$. Similar to the proof of \Cref{thm:pkco} we define $H_e := \{f \in X: e \in p_f\}$ and set $y$ as follows:
\begin{equation*}
    y_e := \sum_{f \in H_e} x_f.
\end{equation*}
Now, we argue that $y$ is a feasible solution for \ref{lp:wfpp} with objective value at least $m$. The flow is conserved for each vertex $v \in V \cup X$ since for any $f \in X$, we add the same amount $x_f$ to $y_e$ of all $e \in p_f$. Observe that $\flo(v) = \cov(v)$ thus the constraint $\cov(v) \leq 1$ in \ref{lp:pkco} implies $\flo(v) \leq 1$. To see why the rank constraints are satisfied, the key observation is that any $e \in \din(t)$ must be of the form $(f,t)$ for some $f \in X$, and by our construction $y_e = x_f$. So according to constraint $\sum_{f \in T} x_f \leq \rank(T)$ in \ref{lp:pmco} we have $\sum_{f \in T} y_{(f,t)} \leq \rank(T)$. Lastly, one can follow the argument from the proof of \Cref{thm:pkco} to show that the \ref{lp:wfpp} objective for this solution will be at least $m$.
\end{proof}

\section{Priority Knapsack Center with Outliers}\label{sec:pknapco}
In this section, we discuss the Priority Knapsack Center with Outliers (\pknapco) problem.
\begin{definition}[Priority Knapsack Center with Outliers (\pknapco)]
The input is a metric space $(X,d)$, a radius function $r: X \to \Rset_{>0}$, a weight function $\w: X \to \Rset_{\geq0}$, parameters $B > 0$ and $m \in \Nset$. The goal is to find $S \subseteq X$ with $\w(S) \leq B$ to minimize $\alpha$ such that for at least $m$ points $v \in X$, $d(v,S) \leq \alpha \cdot r(v)$.
\end{definition}
 \begin{theorem}\label{thm:pknapco}
	There is a 14-approximation for \pknapco.
\end{theorem}

As in the previous sections, we assume $\alpha = 1$ and work with the feasibility version of the problem. We reduce $\pknapco$ to the following path packing problem. Notice that this definition is nearly the same as that of $\wfpp$ in that we need a specified $\cY$ to describe points close to each vertex of the input DAG. In place of the matroid constraint, here we have a knapsack constraint.

\begin{definition}[Weighted Knapsack Path Packing (\wknappp)]
    The input is DAG $G=(V,E)$, $\lambda$, finite set $X$, and $\cY = \{ Y_v \subseteq X: v \in V\}$ (as in $\wfpp$), as well as weight function $\w: X \rightarrow \Rset_{\geq 0}$ and parameter $B > 0$. The goal is to find a set of disjoint paths $P \subseteq \cP(G)$ with maximum $\val(P)$ for which there exists $S \subseteq X$ with $w(S) \leq B$ such that $\forall p \in P$, $S \cap Y_{\sink(p)} \neq \emptyset$.
\end{definition}


There are two main issues in generalizing our techniques from~\Cref{sec:pkco} and~\Cref{sec:pmco} to handle the knapsack constraint. First, the \wknappp problem seems hard on a general DAG. To circumvent this, we make two changes to the LP-aware \pknapco to \wknappp reduction. First, we modify~\Cref{alg:HS} so that a representative captures points at {\em larger} distances. To be precise, for a representative $u$, \Cref{ln:chld} is modified to:
$D(u) \leftarrow \{v \in U: d(u,v) \leq r_u + 2r_v\}$. Second, the partition induced by $C_i$'s in~\Cref{alg:pkco} is done via powers of $4$ instead of $2$. This is what bumps our approximation factor from 9 to 14. However, this dilation allows the resulting \contact to be a directed out-forest. It is not too hard to solve \wknappp when $G$ is a directed-out forest using dynamic programming (details of this DP algorithm can be found in \Cref{sec:knapsackwfpp}).

The second issue, however, is more serious: After reducing to \wknappp we cannot guarantee that the path packing problem has a good integral solution using $\cov$ from the natural \pknapco LP solution. This is because the natural LP relaxation for \pknapco has an unbounded integrality gap (even in the single radius case) \cite{CLLW13}.

We circumvent this by using the round-or-cut framework of \cite{CN18}. Instead of using the \pknapco LP, we would use $\cov$ in the convex hull of the integral solutions (call it $\CovP$). Of course, we do not know the integral solutions and there may indeed be exponentially such solutions. So, we have to employ the ellipsoid algorithm. In each iteration of ellipsoid, we get some $\cov$ that may or may not be in $\CovP$. In any case, if we manage to get a good path packing solution using this $\cov$, we get an approximate \pknapco solution and we are done. Otherwise, we are able to give ellipsoid a linear constraint that should be satisfied by any point in $\CovP$ but is violated by the current $\cov$. Ultimately, either we find an approximate solution for \pknapco along the way, or ellipsoid prompts that $\CovP$ is empty, indicating that the problem is infeasible.

From here on, let $\cF$ be the set of all possible centers that fit in the budget. That is, $\cF := \{S \subseteq X: \w(S) \leq B\}$. The following is the convex hull of the integral solutions for \pknapco.

\begin{alignat}{4}
\CovP = \{(\cov(v): v \in X) : 
&& \sum_{v \in X} \cov(v) & \geq & ~~m \tag{$\CovP$.1} \label{eq:P1} \\
\forall v\in X, && ~~\cov(v) &:=& \sum_{\substack{S \in \F: \\ d(v,S) \leq r_v}} z_S \tag{$\CovP$.2} \label{eq:P2} \\
&& 	\sum\limits_{S \in \F} z_S & = & ~~ 1  \tag{$\CovP$.3} \label{eq:P3} \\
\forall S\in \F,  && z_S &\geq &0\} \notag \tag{$\CovP$.4} \label{eq:P4}
\end{alignat}

Observe that while the polytope has exponentially many auxiliary variables ($z_S$), its dimension is still $|X|$. In the next section, we will describe the whole reduction process.

\subsection{Reduction to \wknappp}\label{subsec:redwknappp}
In this section, we describe how to reduce \pknapco to \wknappp in detail since all the observations from \Cref{subsec:redwkpp} have to be modified to work on a forest instead of \contact.

From here on, we group clients by their neighborhood radius based on powers of 4, rather than 2. Partition $X$ according to each client's neighborhood size into $C_1,\dots,C_t$, where $C_i := \{ v \in X: 4^{i-1} \leq r_v < 4^i\}$ for $i \in [t]$. Let \modHS be the modified \HS algorithm where the construction of $D(u)$ for a representative $u$ (\Cref{ln:chld}) is modified to:
\begin{equation*}
    D(u) \leftarrow \{v \in U: d(u,v) \leq r_u + 2r_v\}.
\end{equation*}
As a result, two statements in ~\Cref{fact:HS} change to the following (recall $\phi$ is the input ordering which we usually initiate to $\cov$):
\begin{fact}\label{fact:modHS}
The following are true for the output of \modHS:
 (a) $\forall u,v \in S, \phi(u) \geq \phi(v):\ d(u,v) > r_u + 2r_v$, and 
 (b)$\forall{u \in S},\forall{v \in D(u)}:\ d(u,v) \leq r_u + 2r_v.$
\end{fact}

\Cref{alg:pknapco} shows the \pknapco to \wknappp reduction.  The algorithm constructs a directed forest called \forest (see \Cref{def:contactpknapco}) as a part of the \wknappp instance definition. We first run \modHS on each $C_i$ to produce the vertices of our \forest.

\begin{algorithm}[!ht]
	\caption{Reduction to \wknappp}
	\label{alg:pknapco}
	\begin{algorithmic}[1]
		\Require \pknapco instance $\cI = $\pknapcoinst and assignment $\{\cov(v) \in \Rset_{\geq 0} :v\in X\}$
		\State $R_i, \{D(u): u\in R_i\} \from \modHS((C_i, d), r, \cov)$ for all $i \in [t]$ \label{ln:knap-HS}
		\State \text{Construct \forest $G = (V, E)$ per \Cref{def:contactpknapco}} 
		\State $\lambda(v) \from |D(v)|$ for all $v \in V$ \label{lin:lam:knap}
  \State $Y_v \from \{u \in X: d(u,v) \leq r_v\}$ for all $v \in V$ \label{lin:yv:knap}
		\Ensure \wknappp instance $(G=(V,E),\lambda,X,\cY = \{ Y_v \mid v \in V\},\w,B)$
	\end{algorithmic}
\end{algorithm}

\begin{definition}[\forest]\label{def:contactpknapco}
Let $R_i \subseteq C_i$, $i \in [t]$ be the set of representatives acquired after running \modHS procedure on $C_i$ according to \Cref{ln:knap-HS} of \Cref{alg:pknapco}. \forest $G = (V,E)$ is a directed forest on vertex set $V = \bigcup_i R_i$ where the arcs are constructed by the as follows: For $u \in R_i$ and $v \in R_j$ where $i > j$, add the arc $(u,v) \in E$ if there exists $f \in X$ such that $d(u,f) \leq r_u$ and $d(v,f) \leq 2r_v$. Next, remove all the forward edges.\footnote{In a DAG, edge $(u,v)$ is a forward edge if there is a path of length two or more in the graph that connects $u$ to $v$.}
\end{definition}
\Cref{fact:arcdist} for \contact translates to the following fact for \forest
\begin{fact}\label{fact:arcdistforest}
If $u \in R_i$, $v \in R_j$, and $(u,v)$ is an arc in \forest, $d(u,v) \leq r_u + 2r_v$.\footnote{In this reduction we change the arc definition for \forest, but keep the setting of $\cY $ the same as it was in the previous section. As a result, an analog to \Cref{rem:arc} will not hold here.}
\end{fact}

A sharp reader may notice that we defined the \forest to be a DAG in \Cref{def:contactpknapco} but referred to it as a forest of directed rooted trees elsewhere. Indeed one can show that the \forest cannot have any cross edges so one can ignore the directions on the edges. We show this in the following claim.
\begin{claim}
For $u \in R_i$ and $v \in R_j$, there is no $w \in R_k$ such that both $(u,w)$ and $(v,w) \in E$.
\end{claim}

\begin{proof} Assume for the sake of contradiction that such a $u,v,w$ exists, hence $k < j$ and $k < i$. If $(u,w)$ and $(v,w)$ are both in $E$, then $d(u,w) \leq r_u + 2r_w$ and $d(v,w) \leq r_v + 2r_w$. By triangle inequality, we know that:
\begin{equation*}
    d(u,v) \leq r_u + r_v + 4r_w < r_u + r_v + 4^{k+1},
\end{equation*}
which is less than or equal to both $r_u + 2r_v$ and $r_v + 2r_u$ as $k < i$ and $k < j$. This means that either $(u,v)$ or $(v,u)$ must also be in $E$ which makes $(u,w)$ or $(v,w)$ a forward edge. However, since we have removed all forward edges from $G$, we reach a contradiction. 
\end{proof}

The following lemma is what gives the approximation factor stated in \Cref{thm:pknapco}.
\begin{lemma}\label{lma:pknapco} Any solution with value at least $m$ for the output of \Cref{alg:pknapco} translates to a $14$-approximation for the input $\cI$.
\end{lemma}
Before proving the lemma, we will make the following observation on \Cref{alg:pknapco}.

\begin{claim}\label{clm:dist2rootforest}
For any $u \in R_i$ and $v \in R_j$ reachable from $u$ in \forest $G$ and any $f \in Y_v$, $d(u,f) < 2\cdot 4^i$.
\end{claim}
\begin{proof}
By definition of $G$, it must be the case that $j < i$. Also for all $f \in Y_v$, $d(f,v) \leq r_v$. If $v$ is reachable from $u$, a path between $u$ and $v$ may contain a vertex $w_k$ from every level $R_k$ for $ j < k < i$:
\begin{align*}
d(u,f) & \leq d(u,v) + d(v,f) \leq d(u,v) + r_v &\\
&\leq (r_u + 2r_{w_{i-1}}) + (r_{w_{i-1}} + 2r_{w_{i-2}}) + \hdots + (r_{w_{j+1}} + 2r_v) + r_v&\text{(by \Cref{fact:arcdistforest})}\\ 
&\leq r_u + 3\sum_{k = 1}^{i-1} 4^k & \text{(by definition of $C_k$)}\\ 
&= r_u + 3\cdot ((4^i-1)/3-1) < 2\cdot4^i. &\text{(since $u \in C_i$, $r_u < 4^i$)}
\end{align*}
\end{proof}

We are now armed with all the facts needed to prove \Cref{lma:pknapco}.
\begin{proof}[Proof of \Cref{lma:pknapco}]
Let $P \subseteq \cP(G)$ be the promised \wknappp solution. Let $S \in \cF$ be the set that intersects $Y_{\sink(p)}$ for all $p \in P$. We will show that this $S$ is a $14$-approximate solution for the initial \pknapco instance $\cI$. For any $p \in P$ with sink node $v = \sink(p)$ and any $u \in p$, $w \in D(u)$ is covered by dilation at most 14 through $f \in S \cap Y_v$. Assume $u \in R_i$ for some $i \in [t]$.
\begin{align*}
    d(w,f) &\leq d(w,u) + d(u,f)\\
    &< 2r_w + r_u + 2\cdot 4^i  & \text{(by~\Cref{fact:modHS} and \Cref{clm:dist2rootforest})}\\
    &< 2r_w + 3\cdot 4^i & \text{($u \in C_i$ so $r_u < 4^i$)}\\
    &\leq 14r_w. & \text{($w \in C_i$ so $4^{i-1} \leq r_w$)}
\end{align*}
What remains is to show that $m$ points will be covered by $S$. The proof is identical to that in the proof of \Cref{lma:pkco}.
\end{proof}

In the next section, we show how to separate $\cov$ from $\CovP$ if it is not valuable. That is, if after solving the output of \Cref{alg:pknapco} the output \wknappp does not have a solution with value at least $m$, we can prove that the input $\cov$ is not in $\CovP$.

\subsection{The Round or Cut Approach}\label{subsec:cut}

Given  \pknapco instance $\cI$ and $\{cov(v) : v \in X\}$, let $\cJ$ be the \wknappp instance output by \Cref{alg:pknapco} on this input.
We say $\cov$ is \emph{valuable} if $\cJ$ has a solution with value at least $m$. In this section, we will show how to prove $\cov \notin \CovP$ if $\cov$ is not valuable.

The following lemma from \cite{CN18} gives a valid set of inequalities that any point in $\CovP$ has to satisfy.

\begin{lemma}[from \cite{CN18}]\label{lma:farkas}
Let $\lambda(v) \in \Rset$ for every $v \in X$ be such that
\begin{equation}\label{lma:farkas:solution}
    \sum_{\substack{v \in X:\\ d(v,S) \leq r_v}} \lambda(v) < m \tab \forall S \in \F,
\end{equation}
Then any $\cov \in \CovP$ satisfies 
\begin{equation}\label{lma:farkas:lambda}
    \sum_{v \in X} \lambda(v) \cov(v) < m.
\end{equation}
\end{lemma}







\begin{lemma}\label{lma:separate}
If a given \cov is not valuable, there is a hyperplane that separates it from \CovP.
\end{lemma}

To prove this lemma, we need to first prove \Cref{clm:goodlambda} and \Cref{clm:represented}.
\begin{claim}\label{clm:goodlambda}
If $\sum_{v \in X} \cov(v) \geq m$ then $\lambda$ as defined in  \Cref{alg:pknapco} satisfies \[\sum_{v \in X} \lambda(v) \cov(v) \geq m.\]
\end{claim}

\begin{proof}
\begin{align*}
    \sum_{v \in X} \lambda(v) \cov(v) &= \sum_i\sum_{v \in R_i} |D(v)|\cov(v) & \text{(by definition of $\lambda$)}\\
    &\geq \sum_i\sum_{v \in R_i}\sum_{u \in D(v)}\cov(u) &\text{(by~\Cref{fact:HS})}\\
    &= \sum_{u \in X} \cov(u). &\text{(by \Cref{fact:Vpart})}
\end{align*}
\end{proof}

\begin{claim}\label{clm:represented}
For any solution $S \in \cF$, there is a solution $P$ for the output of \Cref{alg:pknapco} such that \[\val(P) = \sum\limits_{\substack{v \in X:\\ d(v,S) \leq r_v }} \lambda(v).\]
\end{claim}

\begin{proof}
Observe that since $\lambda$ is only non-zero for representative points: \[ \sum\limits_{\substack{v \in X:\\ d(v,S) \leq r_v }} \lambda(v) = \sum\limits_{\substack{v \in V:\\ d(v,S) \leq r_v }} \lambda(v). \]

For any $f \in S$, let $A_f := \{v\in V: d(v,f) \leq r_v \}$. Assume $A_f$'s are disjoint by assigning each covered point to an arbitrary facility in $S$ that covers it. So, we have that \[ \sum\limits_{\substack{v \in V:\\ d(v,S) \leq r_v }} \lambda(v) = \sum_{f \in S} \sum_{v \in A_f}\lambda(v). \] 

Fix some $f \in S$, and sort the members of $A_f$ according to the topological order in the \forest. Let $u \in R_i$ and $v \in R_j$ be two consecutive
members of $A_f$ in this order. Observe that the \forest must contain an edge from $u$ to $v$ since $d(u,v) \leq d(u,f) + d(f,v) \leq r_u + r_v$. Hence $A_f$ corresponds to a path in \forest (call it $p_f$). Since $A_f$'s are disjoint $p_f$'s are also vertex disjoint. Also every path sink is covered by $S$. That is, for any $p_f$ and $v = \sink(p_f)$ there is a member of $S$ (that is $f$) at distance at most $r_v$ from $v$. So the set $P = \{A_f \mid f \in S\}$ satisfies all the requirements of a feasible \wknappp solution. Furthermore, 
\[\val(P) = \sum_{p \in P}\sum_{v \in p} \lambda(v) = \sum_{f \in S} \sum_{v \in A_f}\lambda(v) .\]
\end{proof}


Now the proof of \Cref{lma:separate} follows easily.
\begin{proof}[Proof of \Cref{lma:separate}]
Through \Cref{clm:represented} we established the fact that for any $S \in \F$, there exists a \wknappp solution $P$ such that: \[\val(P) = \sum\limits_{\substack{v \in X:\\ d(v,S) \leq r_v }} \lambda(v).\] But by the assumption that $\cov$ is not valuable $\val(P) < m$ hence: \[ \sum\limits_{\substack{v \in X:\\ d(v,S) \leq r_v }} \lambda(v) < m.\]
This contradicts \Cref{lma:farkas} since $\sum_{v \in X} \lambda(v) \cov(v) < m$ is violated per \Cref{clm:goodlambda} and can be used as a separating hyperplane. (Note, we can assume $\sum_{v \in X} \cov(v) \geq m$ since if not, this itself can be used to separate $\cov$).
\end{proof}

With this we complete the description of our round-or-cut approach.
\begin{proof}[Proof of \Cref{thm:pknapco}]
Given \pknapco instance $\cI$, we can produce either a $14$-approximate solution or prove that it is infeasible. We start an ellipsoid algorithm on $\CovP$. Given $\cov$ from ellipsoid, we feed $\cI$ and $\cov$ to \Cref{alg:pknapco} to get a \wknappp instance $\cJ$. Now we use the dynamic program in \Cref{sec:knapsackwfpp} to solve it. If the solution has value at least $m$, by \Cref{lma:pknapco} we have a $14$-approximate solution for $\cI$ and are done. If not, we say $\cov$ is not valuable. Then, \Cref{lma:separate} gives us a separating hyperplane from the ellipsoid as proof that $\cov \notin \CovP$. After polynomially many iterations of ellipsoid, either we get an approximate solution for $\cI$, or ellipsoid prompts that $\CovP$ is empty and $\cI$ is infeasible.
\end{proof}

\section{Connections to Fair Clustering}\label{sec:fair}
In this section, we show how our results imply results in the two fairness notions as defined by~\cite{JKL20} and~\cite{HLPST19}.

\subsection{``A Center in your Neighborhood'' notion of~\cite{JKL20}}

\newcommand{\nr}{\text{NR}}
Jung et al.~\cite{JKL20} argue that fairness in clustering should take
into account population densities and geography. For every $v\in X$,
they define a {\em neighborhood radius} $\nr(v)$ to be the distance to
its $(\lceil n/k \rceil -1)$th nearest neighbor. They argue that a solution is fair if every $v$ is served within their $\nr(v)$. They also
observe that this may not always be possible, and therefore they wish
to find a placement of centers $S\subseteq X$ minimizing
$\max_v \frac{d(v,S)}{\nr(v)}$. As an optimization problem, their problem
is precisely an instantiation of \pkc. Thus, one can easily
obtain a $2$-approximation when we set $r(v) = \nr(v)$.

\cite{JKL20} in fact show that it is always possible to find $S$ such
that $d(v,S) \le 2\nr(v)$. They do so by looking at the centers
obtained from running their algorithm (which is the same as that of
Plesn\'ik). Note that a $2$-approximation to
the instance of \pkc defined by $r(v) = \nr(v)$ does not necessarily
imply this additional property. Here, we show why their finding is not a
coincidence by considering the natural LP relaxation for \pkc.
Given an instance of \pkc one can obtain a lower bound
on the optimum value by finding the smallest $\alpha$ such that
the following LP is feasible.

\begin{equation}
  \mathrm{PkCFeasLP}(\alpha) ~:=~  \{(y_u \geq 0 ~:~u\in X): ~~~\sum_{u\in X}
  y_u \leq k; ~~~~ \forall v\in X: \sum_{u: d(u,v)\leq \alpha r(v)} y_u \geq 1 \} \label{lp:fpkc}
\end{equation}

\begin{claim}
\label{clm:lp}
Suppose $\mathrm{PkCFeasLP}(\alpha)$ has a feasible solution, then~\Cref{alg:HS} run
with $\phi(v) = \frac{1}{r(v)}$ finds at most $k$ centers that
cover each point $v$ within distance $2\alpha r(v)$.
\end{claim}
\begin{proof}
  The proof is similar to that of~\Cref{thm:pkc}. Without loss of
  generality we can assume $\alpha = 1$, otherwise we can
  scale all  the radii by $1/\alpha$.
  We need to argue
  $|S|\leq k$. For any $u\in S$, we have $\sum_{v\in D(u)} y_v \geq 1$
  since $B(u,r(u))\subseteq D(u)$. Since $D(u)$'s are disjoint and
  $\sum_{u\in X}y_u\leq k$, the claim follows.
\end{proof}

The preceding discussion and the claim shows the utility of viewing the
clustering problem of \cite{JKL20} as a special case of \pkc. One
can then bring to bear all the positive algorithmic results on \pkc
(such as \Cref{thm:pks}) to fine-tune the fair clustering model.
Below we list a few other ways in which the \pkc view
could be useful.
\begin{itemize}
\item The LP relaxation could be useful in obtaining better empirical
  solutions. For example, it has been shown that for $k$-Center,
  the LP relaxation is integral under notions of stability~\cite{ChekuriG18}.
\item The model of \cite{JKL20} allows $\nr(v)$ to be very large for
  points $v$ which may not be near many points. However, one may want
  to place an upper bound $M$ on the radius that is independent of
  $\nr(v)$. The same algorithm yields a $2$-approximation
  but one may no longer have the property that all points are
  covered within twice $\nr(v)$.
\item In many scenarios, it makes sense to work with the Supplier
  version since it might not be possible to place centers at all locations in $X$. Second, there could be several additional constraints
  on the set of centers that can be chosen. \Cref{thm:pks} shows
  that more general constraints than cardinality can be handled.
\item As previously mentioned, far away points in less dense
  regions (outliers) can be harmed by setting $\nr(v)$ to be a large
  number. Alternatively, one can skew the choice of centers if one
  tries to set a small radius for these points. In this situation it
  is useful to have algorithms that can handle outliers such that one
  can find a good solution for vast majority of points and help
  the outliers via other techniques.
\end{itemize}

\subsection{The Lottery Model of Harris et al.~\cite{HLPST19}	}\label{sec:lottery}

Harris et al.~\cite{HLPST19} define a lottery model of fairness where every client $v\in X$ has a ``distance demand'' $r(v)$ and a ``probability demand'' $\prob(v)$.
They deem a \emph{lottery}, or distribution, over feasible solutions fair if every client is connected to a facility within distance $r(v)$ with probability at least $\prob(v)$.
The computational question is to figure out if this is (approximately) feasible. We show a connection to the outlier version of the Priority $k$-Center problem, and then generalize their results.
To start, consider the following problem definition.
\begin{definition}[Lottery Priority $\cF$-Center (\lpfc)]
	The input is a metric space $(X,d)$ where each point $v$ has a distance demand $r(v) > 0$ and probability demand $\prob(v)$. 
	The input also (implicitly) specifies a family $\cF \subseteq 2^X$ of allowed locations where centers can be opened. A distribution over $\cF$ is $\alpha$-approximate if
	\[
	\forall v\in X:~~~\Pr_{S\sim \cF} [d(v,S) \leq \alpha \cdot r(v)] \geq \prob(v).
	\]
	An $\alpha$-approximation algorithm in the lottery model either asserts the instance infeasible in that a $1$-approximate distribution doesn't exist, or returns an $\alpha$-approximate distribution.
\end{definition}
Harris et al.~\cite{HLPST19} show that for the case when $\cF$ is $\{S: |S|\le k\}$, there is a $9$-approximate distribution. Using our aforementioned results, and a standard framework based on the Ellipsoid method (as in~\cite{CV02,AAKZ20}), 
we get the following results.
%
%
\begin{theorem} \label{thm:lpmc} There is a 9-approximation for \lpfc where $\cF$ is the independent set of a matroid. \end{theorem}
\begin{theorem} \label{thm:lknapc} There is a 14-approximation for \lpfc on points $X$ where $\cF = \{S \subseteq X: \w(S) \leq B\}$ for a poly-bounded weight function $\w: X \to \Rset_{\geq0}$ and parameter $B > 0$.\end{theorem}

%
%
We first describe the reduction. For this, we need to define the Fractional Priority $\cF$-Center problem, in which each point comes with a (possibly fractional) weight $\mu_v$, and given $m \geq 0$, the goal is to find a set $S \in \cF$ that covers a total weight of more than $m$ with minimum dilation of neighborhood radii. 

\begin{definition}[Fractional Priority $\cF$-Center (\fpfc)]
The input is a metric space $(X,d)$ where each point $v$ has a radius $r_v > 0$ and a weight $\mu_v \geq 0$. Given parameter $m \geq 0$ and a family of subsets of points $\cF \subseteq 2^X$, the goal is to find $S \in \cF$ to minimize $\alpha$ such that $\mu(\{v \in X$: $d(v,S) \leq \alpha \cdot r_v\}) > m$.
\end{definition}
An instance of \fpfc is specified by the tuple $\fpfcinst$. The following theorem states the reduction from \lpfc to \fpfc using the Ellipsoid method. The proof of this theorem can be found in \Cref{sec:thmlpfcfpfc}.


\begin{theorem} \label{thm:lpfc-fpfc}
Given \lpfc instance $\cI$ and a black-box $\alpha$-approximate algorithm $\cA$ for \fpfc that runs in time $T(\cA)$, one can get an $\alpha$-approximate solution for $\cI$ in time $\text{poly}(|\cI|)T(\cA)$.
\end{theorem}

Now we discuss how our results generalize to solve \fpfc for matroid and knapsack constraints.

\begin{proof}[Proof of \Cref{thm:lpmc}]
According to \Cref{thm:lpfc-fpfc} we only need to prove that we can find a 9-approximate solution for any given \fpfc instance $\cI = \fpfcinst$. First, observe that the LP for $\cI$ is the same as \ref{lp:pmco} with a minor modification: The constraint $\sum_{v \in X}\cov(v) \geq m$ is changed to $\sum_{v \in X}\mu_v\cov(v) > m$. Solve the LP for $\cI$ and use the obtained $\cov$ to run the reduction in \Cref{alg:pmco} but with a change in \Cref{lin:lam}: instead of setting $\lambda(v) \from |D(v)|$ for all $v \in V$, we will have $\lambda(v) \from \mu(D(v))$. This results in a \wfpp instance $\cJ$ with fractional $\lambda$. The procedure in \cite{ckrv12} can handle fractional $\lambda$'s so we can still compute the solution for $\cJ$ in polynomial time. If this solution has value less than or equal to $m$, we know that $\cI$ is infeasible. Otherwise, \Cref{lma:pmco} tells us that this solution for $\cJ$ translates to a $9$-approximation for $\cI$ and we are done.
\end{proof}

\begin{proof}[Proof of \Cref{thm:lknapc}]
We follow a procedure similar to the proof of \Cref{thm:lpmc}. Per \Cref{thm:lpfc-fpfc} we only need to prove there is a 14-approximation for the \fpfc instance where $\cF$ is a set of feasible knapsack solutions with poly-bounded weights $\w: X \to \Rset_{\geq0}$ and budget $B > 0$. Change the constraint \ref{eq:P1} in $\CovP$ to $\sum_{v \in X}\mu_v\cov(v) > m$ and modify \Cref{lin:lam:knap} of \Cref{alg:pknapco} to $\lambda(v) \from \mu(D(v))$ then follow the round-or-cut procedure in the proof of \Cref{thm:pknapco}. The only challenge here is to prove the \wknappp problem can be solved in polynomial time. The dynamic program in \Cref{sec:knapsackwfpp} depends on the assumption that $\lambda$'s are poly-bounded. But here, our $\lambda$'s are real numbers so instead, we assume that our weights $\w: X \to \Rset_{\geq0}$ are poly-bounded so we can still solve the problem via dynamic programming.
\end{proof}

\bibliographystyle{alpha}
\bibliography{references}

\appendix

\section{Weighted Knapsack Path Packing}\label{sec:knapsackwfpp}
\begin{claim}\label{clm:pfco:knapsack}
Any \wknappp instance \wknapppinst where $G$ is a forest can be solved in polynomial time.
\end{claim}
\begin{proof}
We overload the weight assignment for $v \in V$ as follows: $\w(v) := \min_{\substack{f \in Y_v}} \w(f)$.
The goal is to find a set of disjoint paths $P \in \cP(G)$ such that $\w(P) := \sum_{p \in P} \w(\sink(p)) \leq B$ and $\val(P):= \sum_{p \in P} \sum_{v \in p}\lambda(v)$ is maximized.
 Equivalently, one can find maximum $m \in \{0,\dots,n\}$ for which there exists such $P$ with $\val(P)= m$ and $\w(P) \leq B$.
 This can be easily done via a dynamic program, since $\lambda$'s are integer valued and poly-bounded.
 
 In \Cref{alg:tree:knap} table $Q$ holds the partial solutions. For $u \in V$ and $l \in \{0,\dots,n\}$, $Q[u][l]$ is the minimum weight of a set of disjoint paths in the sub-tree rooted at $u$ that covers at least $l$ total $\lambda$. For simplicity, assume $G$ is actually a tree\footnote{This is without loss of generality since one can always add a dummy root $r$ with $\lambda(r) = w(r) = 0$.} rooted at $r \in V$, then we are interested in the maximum $m \in \{0,\dots,n\}$ for which $Q[r][m] \leq B$. 
 
 We claim that for any $u \in V$, if $P$ has a sink in the sub-tree rooted at $u$, without loss of generality we can say that $u$ is not a sink in $P$. Suppose otherwise, that $u$ is a sink in $P$ and there exists another sink $v \in V$ in the sub-tree of $u$. Then one can join the path ending at $u$ to the path from $u$ to $v$, while decreasing $\w(P)$ by $\w(u)$ and potentially increasing the total $\lambda$ covered. 
 
 So for any sub-problem $(u,l)$, there is the option of making $u$ be a sink and discarding the whole sub-tree (in which case it should be that $\lambda(u) \geq l$), or dividing the responsibility of covering $l - \lambda(u)$ total $\lambda$ to children of $u$ in the sub-tree (denoted by $\dout(u)$). This itself needs a subset-sum dynamic program to account for different ways $l -\lambda(u)$ can be broken up into $d^+_u$ parts where $d^+_u = \dout(u)$. This other dynamic program is done in table $Z$. For $j \in [d^+_u]$ and $k \in \{0,\dots,n\}$, $Z[j][k]$ is the minimum weight of disjoint paths in the subtrees rooted at the first $j$ children of $u$, covering at least $k$ total $\lambda$ value.\footnote{This is with a slight abuse of notation: In \Cref{alg:tree:knap} the first index of $Z$ starts at 0 even though the children are indexed at 1. The 0 index indicates that we have already gone over all the children.}
 
\end{proof}
\begin{algorithm}[!ht]
	\caption{Knapsack on Trees}
	\label{alg:tree:knap}
	\begin{algorithmic}[1]
		\Require $G = (V, E)$ tree rooted at $r$, assignments $\{\lambda(v) \in \Rset_{\geq 0} :v\in  V\}$ and $\{\w(v) \in \Rset_{\geq 0} :v\in  V\}$
		\State $Q[V][0,\dots,n] \leftarrow \infty$ 
		\For{$u \in V$}
		    \State $Q[u][1,\dots,\lambda(u)] \leftarrow \w(u)$ \Comment{Base case: initialize all these entries to $\w(u)$}
		    \State $Q[u][0] \leftarrow 0$ \Comment{Base case: 0 weight is spent to cover 0 total $\lambda$}
		\EndFor
		\For{ $i=1$ to $t$}
		    \For{$u \in R_i$}
		        \State $D_u[1,\dots,d^+_u] := \text{array of } v \in \dout(u)$ \Comment{$D_u[i]$ is the $i$th child of $u$}
		        \State $Z[0,\dots,d^+_u][0,\dots,n] \leftarrow \infty$
		        \State $Z[0,\dots,d^+_u][0] \leftarrow 0$ \Comment{Base case: 0 weight is spent to cover 0 total $\lambda$}
		        \For {$j = 1$ to $d^+_u$} \Comment{Deciding on $D_u[1,\dots,j]$}
		            \For{$k = 1$ to $n$} \Comment{Trying to cover $k$ total $\lambda$}
		                \For{$l = 0$ to $k$} \Comment{How much $\lambda$ we cover at the sub-tree rooted at $D_u[j]$}
		                    \State $Z[j][k]\leftarrow \min\{Z[j][k], Q[D_u[j]][l] + Z[j-1][k-l]\}$
		                \EndFor
		            \EndFor
		        \EndFor
		        \For{$k = 1$ to $n$} \Comment{Compute $Q[u]$ based off of $Z$}
		            \State $m \leftarrow \min\{n,k+\lambda(u)\}$
		            \State $Q[u][m]\leftarrow \min\{Q[u][m], Z[d^+_u][k]\}$
		        \EndFor
		    \EndFor
		\EndFor
		\State $m^* \from \max\{m \in \{0,\dots,n\}: Q[r][m] \leq B\}$
		\Ensure $m^*$
	\end{algorithmic}
\end{algorithm}

\section{Proof of Theorem~\ref{thm:lpfc-fpfc}}\label{sec:thmlpfcfpfc}
\begin{proof}
	An $\alpha$-approximate solution for \lpfc exists if the following polytope is non-empty. One can think of $z$ as a probability distribution over $\cF$:
	\begin{align}
	\sum_{\substack{S \in \F: \\ d(v,S) \leq \alpha r_v}} z_S &\geq \prob_v & \forall v \in X \tag{\lpfc LP}\label{lp:lpfc} \\
	\sum\limits_{S \in \F} z_S &=1&  \notag \\
	z_S &\geq 0 & \forall S \in \cF.\notag
	\end{align}
	Since the above LP has $|X| + 1$ many non-trivial constraints, any basic feasible solution $z$ has support size at most $|X|+1$. So we know that if the LP is feasible, the $\cH$ and $z$ solution described in the definition of \lpfc do exist.
	
	One could view the above LP as a standard minimization problem. Then the dual LP is a maximization problem with variables $\mu_v$ for all $v \in X$ and a variable $m \in \Rset$.
	
	\begin{align}
	\max \sum_{v \in X}\prob_v \mu_v &- m \tag{\fpfc LP}\label{lp:fpfc} \\
	\sum_{\substack{v \in X: \\ d(v,S) \leq \alpha r_v}} \mu_v& \leq m&  \forall S \in \cF\notag \\
	m &\in \Rset & \notag\\
	\mu_v &\geq 0 & \forall v \in X.\notag
	\end{align}
	Observe that, for all practical purposes we can assume $m \geq 0$ since $\mu_v \geq 0$ and even if one $S\in \cF$ covers a point $v \in X$, that enforces $m \geq 0$. Also the vectors of all zeros is a feasible solution so the objective value is at least 0. One can assume that \ref{lp:lpfc} had a fixed objective of $0$ so \ref{lp:lpfc} is feasible iff the objective value of \ref{lp:fpfc} is 0. That is, any $(\mu,m)$ that satisfies the constraints in \ref{lp:fpfc} but has $\sum_{v \in X} \prob_v\mu_v > m$ certifies \ref{lp:lpfc} is infeasible. Thus\footnote{The reason we changed $> m$ to $\geq m+1$ in \ref{lp:q} is only because we would be using this constraint as a separating hyperplane in ellipsoid so we cannot work with strict inequalities. This is allowed  since \ref{lp:fpfc} is scale invariant. That is, if $(\mu,m)$ satisfies the constraints, so does $(\beta\mu,\beta m)$ for $\beta > 0$ so if its objective value is anything greater than 0, we could scale them properly and assume the objective value is 1.} \ref{lp:lpfc} is feasible iff the following polytope is empty:
	\begin{align}
	\sum_{v \in X}\prob_v \mu_v &\geq m+1 \tag{$L(\alpha)$}\label{lp:q} \\
	\sum_{\substack{v \in X: \\ d(v,S) \leq \alpha r_v}} \mu_v& \leq m&  \forall S \in \cF\notag \\
	m &\geq 0 & \notag\\
	\mu_v &\geq 0 & \forall v \in X.\notag
	\end{align}
	
	Now we will be running an ellipsoid on \ref{lp:q}. Through running ellipsoid, we either find a feasible solution for \ref{lp:q} which certifies that our \lpfc instance does not have an $\alpha$-approximation, or we prove through a set of separating hyperplanes that \ref{lp:q} is empty. Then, these hyperplanes will be used to construct our \lpfc solution.
	
	Start with $\cH = \emptyset$. Ellipsoid gives $(\mu,m)$ and queries whether it is in \ref{lp:q} or not. We check if any one of the constraints $\sum_{v \in X}\prob_v \mu_v \geq m+1$ or $m \geq 0$ or $\mu_v \geq 0$ for a $v \in X$ are violated. If so, we return it to ellipsoid as a separating hyperplane. Otherwise, we run algorithm $\cA$ on the \fpfc instance $\cJ = \fpfcinst$. If $\cA$ returns an $\alpha$-approximate solution $S \in \cF$ for $\cJ$, add $S$ to $\cH$. The constraint $\sum_{\substack{v \in X: \\ d(v,S) \leq \alpha r_v}} \mu_v \leq m$ is violated for this $S$ and is fed to ellipsoid as a separating hyperplane.
	
	If $\cA$ fails, we know that $\cJ$ is infeasible, that is, no $S \in \cF$ can cover more than $m$ weight of points. So $(\mu,m)$ is a feasible solution for $L(1)$ which certifies $\cI$ does not have a solution with dilation 1. So we terminate the procedure.
	
	If ellipsoid decides that \ref{lp:q} is empty, we solve \ref{lp:lpfc} projected only on sets $S \in \cH$ and find our probability distribution $z$ on $\cH$. Note that $\cH$ is of polynomial size because ellipsoid terminates in polynomially many iterations and we generate at most one member of $\cH$ in each iteration.
\end{proof}

\section{Handling Priority Supplier with Outliers}\label{sec:pkso}

 In this section we discuss how our framework can be used to handle the more general Supplier versions of \pkco, \pmco, and \pknapco. For convenience, we provide the following definition for the Priority Supplier with Outliers problems. We consider the feasibility problem for simplicity.

\begin{definition}[Priority $\F$-Supplier with Outliers] 
  The input is a metric space $(X,d)$ where $X = F \uplus C$, $C$ is the
  set of clients, and $F$ the set of facilities. In addition there is a
  radius function $r: C \to \Rset^+$, a parameter $m \in \Nset$, and a down-ward closed family $\cF$ of subsets of $F$.  The goal is to find
  $S \in \cF$ such that for at least $m$ clients $v \in C$,
  $d(v,S) \leq r(v)$. 
\end{definition}


Different settings of $\cF$ lead to different problems.  We obtain the Priority $k$-Supplier with Outliers problem if $\F = \{S \subseteq F \mid |S| \leq k\}$.  We obtain the Priority Matroid Supplier with Outliers problem when $(F,\F)$ is a matroid.  We obtain the Priority Knapsack Supplier with Outliers problem when there is a weight function $w:F\to \Rset_{\geq 0}$ and $\F = \{S \subseteq F \mid w(S) \leq B\}$ for
some budget $B$.

The following is the natural LP relaxation for the feasibility version of Priority $k$-Supplier with Outliers. For each \emph{facility} $u \in F$, there is a variable $x_u \in [0,1]$ that denotes the fractional amount by which $u$ is opened as a center. For each \emph{client} $v \in C$, the quantity $\cov(v)$ is used to indicate the amount by which that client is covered by facilities within distance $r_v$ from $v$. We want to make sure that at least $m$ units of coverage are assigned to clients using at most $k$ facilities.

\begin{align}
\sum_{v \in C}\cov(v) &\geq m\tag{Priority $k$-Suppliers with Outliers LP}\label{lp:pkso} \\
\sum_{u \in F} x_u &\leq k \tag{$*$} \\ \\
\cov(v) := \sum_{\substack{u \in F : \\ d(u,v) \leq r_v}}  x_u &\leq 1 \qquad \forall v \in C \notag \\
0 \leq x_u &\leq 1 \qquad \forall u \in F\notag
\end{align}

For Priority Matroid Supplier with Outliers, the constraint $(*)$ will be replaced by the matroid rank constraint on the facilities:$\sum_{u \in S} x_u \leq \rank(S),~\forall S \subseteq F$. We can similarly write the convex hull of integral solutions for Priority Knapsack Supplier with Outliers (as in \Cref{sec:pknapco}) by taking care to differentiate between facility points and client points. 

We will first describe the changes needed to apply our framework to Priority $k$-Suppliers with Outliers. Using the $\cov$ values from the above LP, we run the filtering procedure and reduction from \Cref{sec:pkco} on the client set $C$ (as opposed to the entire set $X$); recall that this reduction only utilizes the $\cov$ values of the LP, and applies when restricting to the point set to $C$. Thus, the \contact we construct will be have vertices that represent clients, and the solution to \wkpp will pick $k$ clients as path sinks, i.e. $S := \{ \sink(p) \mid p \in P\}$ where $P$ is the set of $k$ disjoint paths chosen as the solution of value $m$ for the \wkpp instance. Note that $S \subseteq C$. We obtain a set of facilities $S' \subseteq F$ as follows: For each $v \in S$, add to $S'$ an arbitrary $f \in F$ such that $d(v,f) \leq r_v$. Such an $f$ will exist for each client in $S$, else the LP-relaxation would be infeasible. The algorithm outputs $S'$. It is easy to see that $|S'| \le k$.

The analysis for \pkco shows that $S$ yields a $9$-approximation.
It may appear that choosing $S'$ instead of $S$ will incur an additional $+1$ in approximation. However, we argue that $S'$ is a $9$-approximate solution. This is because the analysis for \Cref{clm:dist2root} has some slack. In fact, we have already utilized this slack in \Cref{sec:pmco} --- see the proof of \Cref{clm:dist2root} for more details. The results from \Cref{thm:2pkc,thm:bpkc} also hold for the supplier version. In particular, for the special case where there are exactly $t \geq 2$ distinct radii types, the algorithm described in the proof of \Cref{thm:2pkc} already describes picking an $f \in X$ that certain clients in the resulting paths, and that is factored into the dilation analysis and will be identical to the analysis for the supplier version. For the special case where radii are powers of some $b \geq 2$, the slack from \Cref{clm:dist2root} will once again allow us to incur an addition cost from choosing a facility closest to the sink client of those paths.

For Priority Matroid and Knapsack Supplier with Outliers, a similar strategy can be used. More precisely, first run the procedure and reduction from \Cref{sec:pmco} and \Cref{sec:pknapco}, respectively, on just the set of clients. We need to slightly alter the definitions of \wfpp and \wknappp. We will alter only \Cref{def:wfpp}, since \wfpp and \wknappp are special cases of Weighted $\cF$-Path Packing. We now have finite set $X = F \uplus C$, $\cF \subseteq 2^F$ and $\cY = \{Y_v \subseteq F \mid v \in V\}$. To adjust our algorithms for the supplier versions, change \Cref{lin:yv} of \Cref{alg:pmco} and \Cref{lin:yv:knap} of \Cref{alg:pknapco} to $Y_v \from \{u \in F \mid d(u,v) \leq r_v \}$ for all $v \in V$. We can once again guarantee that each $Y_v$ is non-empty, since if not (in the Matroid case) the LP will be infeasible or (in the Knapsack case) the convex hull of integral points would be empty.  After this change to $\cY$, the remainder of each algorithm will stay the same. That is, the solution set $S$ chosen (described in \Cref{lma:pmco,lma:pknapco}) will be a subset of facilities that belong to $\cF$. Thus, the results from \Cref{sec:pmco,sec:pknapco} can be obtained in the supplier setting.

\end{document}